\documentclass[10pt,a4paper,reqno,final]{amsart}

\usepackage[T1]{fontenc}
\usepackage[utf8]{inputenc}
\usepackage{lmodern}
\usepackage{amsthm}
\usepackage{amsmath}
\usepackage{amssymb}
\usepackage{mathrsfs}
\usepackage{graphicx,color}
\usepackage{hyperref}

\usepackage[normalem]{ulem}
\usepackage{cancel}

\theoremstyle{plain}
\newtheorem{thm}{Theorem}
\newtheorem{prop}[thm]{Proposition}
\newtheorem{lem}[thm]{Lemma}
\newtheorem{cor}[thm]{Corollary}

\theoremstyle{definition}

\theoremstyle{remark}
\newtheorem{rem}[thm]{Remark}
\newtheorem*{rem*}{Remark}

\newcommand{\N}{\mathbb{N}}
\newcommand{\R}{{\mathbb{R}}}
\newcommand{\C}{{\mathbb{C}}}

\newcommand{\Z}{{\mathbb{Z}}}
\newcommand{\dd}{{\rm d}}
\newcommand{\ii}{{\rm i}}

\newcommand{\diag}{\mathop\mathrm{diag}\nolimits}

\newcommand{\spec}{\mathop\mathrm{spec}\nolimits}

\renewcommand{\Im}{\mathop\mathrm{Im}\nolimits}
\newcommand{\supp}{\mathop\mathrm{supp}\nolimits}

\newcommand{\Res}{\mathop\mathrm{Res}\nolimits}

\newcommand*\pFqskip{8mu}
\catcode`,\active
\newcommand*\pFq{\begingroup
        \catcode`\,\active
        \def ,{\mskip\pFqskip\relax}%
        \dopFq
}
\catcode`\,12
\def\dopFq#1#2#3#4#5{%
        {}_{#1}\phi_{#2}\biggl(\genfrac..{0pt}{}{#3}{#4}\biggl|q;#5\biggr)%
        \endgroup
}

\title[Symmetric orthogonal polynomials and a kinetic spin chain]{New family of symmetric orthogonal polynomials and a solvable model of a kinetic spin chain}

\author{Tomáš Kalvoda}
\address[Tomáš Kalvoda]{
	Department of Applied Mathematics, Faculty of Information Technology, Czech Technical University in~Prague, 
	Th{\' a}kurova~9, 160~00 Praha, Czech Republic
	}	
\email{tomas.kalvoda@fit.cvut.cz}

\author{František Štampach}
\address[Franti{\v s}ek {\v S}tampach]{
	Department of Mathematics, Faculty of Nuclear Sciences and Physical Engineering, Czech Technical University in Prague, Trojanova~13, 12000 Praha~2, Czech Republic
	}
\email{stampfra@fjfi.cvut.cz}

\date{\today}

\begin{document}

\begin{abstract} 
  We study an infinite one-dimensional Ising spin chain where each particle interacts only with its nearest neighbors and is in contact with a heat bath with temperature decaying hyperbolically along the chain.
  The time evolution of the magnetization (spin expectation value) is governed by a semi-infinite Jacobi matrix. The matrix belongs to a three-parameter family of Jacobi matrices whose spectral problem turns out to be solvable in terms of the basic hypergeometric series. As a consequence, we deduce the essential properties of the corresponding orthogonal polynomials, which seem to be new. Finally, we return to the Ising model and study the time evolution of magnetization and two-spin correlations.
\end{abstract}

\maketitle

%
%
\section{Introduction}
\label{sec:Intro}

Among exactly solvable models of statistical physics, the kinetic Ising model plays a prominent role.
The particular case of the one-dimensional kinetic Ising spin chain consists of $N$ particles labeled by integers $\{1,2,\dots,N\}$.
Each particle $n$ is occupied by a spin $\sigma_{n}\in\{\pm1\}$, interacts with its nearest neighboring particles, and is in contact with a heat bath of temperature $T_{n}$, for a simple graphical illustration see Figure~\ref{fig.ising}.
The classical case of constant temperature $T_{n} = T > 0$ was examined by Glauber in \cite{glauber_jmp63}.

\begin{figure}
  {\centering
  \includegraphics{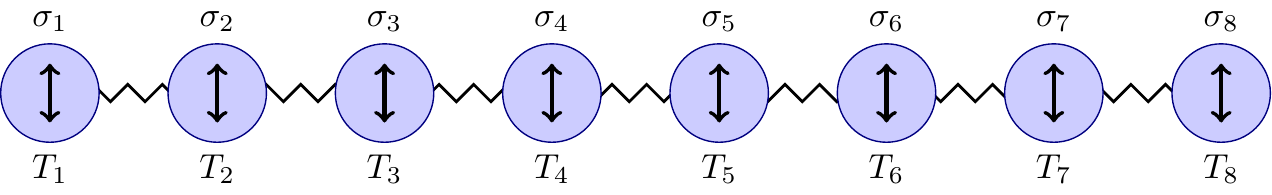}
  \par}
  \caption{Graphical representation of one-dimensional kinetic Ising spin chain with $N=8$ particles.}
  \label{fig.ising}
\end{figure}

The configuration of the system is a vector $\sigma=(\sigma_{1},\dots,\sigma_{N})\in\{\pm1\}^{N}$.
Let $w_{n}(\sigma_{n})$ denote the probability per unit time that the $n$th spin flips from the value $\sigma_{n}$ to $-\sigma_{n}$ and $p(\sigma;t) = p(\sigma_{1},\dots,\sigma_{N};t)$ stand for the probability that the system is in configuration $\sigma$ at the time $t$.
The time evolution of $p(\sigma;t)$ is governed by the master equation
\begin{equation}
  \label{eq:master}
  \frac{\dd}{\dd t}p(\sigma;t)=-\left(\sum_{n=1}^{N}w_{n}(\sigma_{n})\right)p(\sigma;t)+\sum_{n=1}^{N}w_{n}(-\sigma_{n})p(\sigma_{1},\dots,-\sigma_{n},\dots,\sigma_{N};t).
\end{equation}
The transition probabilities $w_{n}(\sigma_{n})$ are chosen to be given by the generalized Glauber transition rate (cf.~\cite{Fonseca2015}):
\begin{equation}
  \label{eq:glauberTransition}
  w_{n}(\sigma_{n})=\frac{1}{2}-\frac{1}{4}\gamma_{n}\sigma_{n}\left(\sigma_{n-1}+\sigma_{n+1}\right),
\end{equation}
for $n\in\{1,\dots,N\}$, where one has to put $\sigma_{0} = \sigma_{N+1} := 0$.
The factor $\gamma_{n}$ is determined by the temperature as
\begin{equation}
  \label{eq:gamma_n}
  \gamma_{n}=\tanh\left(\frac{2}{k_{B}T_{n}}\right), \quad n\in\N,
\end{equation}
where $k_{B}$ is the Boltzmann constant and $\N$ is the set of natural numbers (in our convention $\N$ does not contain $0$; $\N_0 := \N \cup \{0\}$).

The probability functions $p(\sigma;t)$ which satisfy the master equation~\eqref{eq:master} provide the fullest possible description of the system. Nevertheless, as pointed out in \cite{glauber_jmp63}, they contain vastly more information than is needed in practice.
The key macroscopic observables of interest are the expectation values of a spin $\sigma_{n}(t)$ regarded as the stochastic function of time,
\begin{equation}
  \label{eq:magnetization}
  q_{n}(t):=\langle \sigma_{n}(t) \rangle =\sum_{\sigma\in\{\pm1\}^{N}}\sigma_{n}p(\sigma;t),
\end{equation}
and the average of the product of a pair of spins,
\begin{equation}
  \label{eq:2SpinCorrelation}
  r_{m,n}(t):=\langle \sigma_{m}(t)\sigma_{n}(t) \rangle=\sum_{\sigma\in\{\pm1\}^{N}}\sigma_{m}\sigma_{n}p(\sigma;t).
\end{equation}
The quantity $q_{n}(t)$ is called the \emph{magnetization} of the $n$th particle and $r_{m,n}(t)$ is called the \emph{two-spin correlation} of the $m$th and $n$th particle.
An essential physical question asks for the probability that an individual spin or a pair of spins occupy specified states that can be computed if the magnetization and two-spin correlation are known, see \cite[Eqs.~(25) and~(26)]{glauber_jmp63}.

Multiplying by $\sigma_n$ and averaging the both sides of~\eqref{eq:master}, one gets the equation of motion for the magnetization
\begin{equation}
  \label{eq:magnetizationODE}
  \dot{q}_{n}(t) = -q_{n}(t)+\frac{\gamma_n}{2}\left(q_{n-1}(t) + q_{n+1}(t)\right), \quad n\in\{1,\ldots,N\},
\end{equation}
where $q_{0}(t) = q_{N+1}(t) := 0$.
Similarly, multiplying by $\sigma_m\sigma_n$ ($m \neq n$) and averaging~\eqref{eq:master}, one arrives at the equation of motion for the two-spin correlation
\[
  \dot{r}_{m,n}(t) = - 2r_{m,n}(t) + \frac{\gamma_m}{2} \left(r_{m+1,n} + r_{m-1,n}\right) + \frac{\gamma_n}{2} \left(r_{m,n+1} + r_{m,n-1}\right),
\]
where $m,n\in\{1,\ldots,N\}$, $m \neq n$, and $r_{0,n}(t) = r_{m,0}(t) = r_{N+1,n}(t) = r_{m,N+1}(t) := 0$.
Of course, for $m = n$ we have $r_{m,m}(t) = 1$.

The system~\eqref{eq:magnetizationODE} governing the magnetization can be easily expressed in matrix form
\[
  \dot{q}(t) = \Gamma_{N} J_N \Gamma_{N}^{-1} q(t),
\]
where $q(t) := \left(q_{1}(t),\ldots,q_{N}(t)\right)^{T}$, $\Gamma_{N} := \diag\left(\sqrt{\gamma_1},\ldots,\sqrt{\gamma_N}\right)$, and
\[
  J_{N} := \frac{1}{2}\begin{pmatrix}
    -2 & \sqrt{\gamma_{1}\gamma_{2}} & \\
    \sqrt{\gamma_{1}\gamma_{2}} & -2 & \sqrt{\gamma_{2}\gamma_{3}} \\
    & \ddots & \ddots & \ddots \\
    & & \ddots & \ddots & \ddots \\
    & & & \sqrt{\gamma_{N-2}\gamma_{N-1}} & -2 & \sqrt{\gamma_{N-1}\gamma_{N}} \\
    & & & & \sqrt{\gamma_{N-1}\gamma_{N}} & -2 
   \end{pmatrix}\!.
\]
The solution of the system~\eqref{eq:magnetizationODE} then reads
\begin{equation}
  \label{eq:magnetizationODEsolution}
  q(t) = \Gamma_{N} \exp(tJ_{N}) \Gamma_{N}^{-1} q(0).
\end{equation}

\begin{rem}
  Such a symmetric formulation of the problem can be also obtained if one starts with generalized Glauber transition rate
  \[
    w_{n}(\sigma_{n})=\frac{1}{2}-\frac{1}{4}\sqrt{\gamma_{n}}\sigma_{n}\left(\sqrt{\gamma_{n-1}}\sigma_{n-1}+\sqrt{\gamma_{n+1}}\sigma_{n+1}\right).
  \]
  instead of~\eqref{eq:glauberTransition}.
\end{rem}

Since eigenvalues of a non-decomposable self-adjoint Jacobi matrix are always simple, the exponential matrix $\exp(tJ_{N})$ is completely determined by the eigenvalues and eigenvectors of $J_{N}$; namely,
\[
  \exp(tJ_{N})=\sum_{n=1}^{N}e^{t\lambda_{n}}\langle u_{n}, \cdot \rangle u_{n},
\]
where $u_{n}$ is the normalized eigenvector of $J_{N}$ corresponding to the eigenvalue $\lambda_{n}$.
Thus, one is naturally led to the spectral analysis of the Jacobi matrix $J_{N}$.
If the number of particles $N$ is assumed to be large, it might be convenient to replace $J_{N}$ by the corresponding semi-infinite Jacobi matrix and analyze spectral properties of the respective Jacobi operator $J$, see \cite{glauber_jmp63}.

The goal of this paper is to extend the class of solvable models of one-dimensional kinetic Ising chain in the spirit of Glauber. It means that we focus on a rigorous derivation of formulas for the key macroscopic observables in terms of special functions and their asymptotic behavior. The particular model under investigation assumes the temperature decreases hyperbolically along the chain, i.e., $T_{n} = c/n$, $c > 0$.
In this case,
\begin{equation}
  \label{eq:gamma_tanh}
  \gamma_{n} = \tanh(\kappa n), \quad n\in\N,
\end{equation}
where $\kappa>0$, by~\eqref{eq:gamma_n}. Studying this model, special functions from the $q$-world~\cite{gasper90} naturally appear and play an essential role.
It is worth noticing that the analysis of this model, which might be of physical relevance, relies very much on particular aspects of the theory of $q$-hypergeometric series with the quantum parameter~$q$ determined by physical constants.
Such a fruitful intersection of two typically separated areas seems to be quite exceptional in the context of statistical physics. 

Since the spectral properties of tridiagonal matrices determine the dynamics of the magnetization, there is also a close relation to the theory of orthogonal polynomials. Exploring the model in question, we realized that we could provide an analysis of properties of the corresponding family of orthogonal polynomials in even higher generality. Therefore we first deduce properties of a three-parameter family of orthogonal polynomials that seem to be new and are of independent interest. Then we apply these results in a more concrete form to the model under investigation.

To mention more recent works devoted to solvable models of one-dimensional kinetic Ising chain with varying temperature, we refer to papers~\cite{bauer-cornu_jpa18, borchers-etal_pre14, mazilu-etal_ejla12, mazilu-williams_pre09, mobilia-etal_pre05, mobilia-etal_jpa04} that mostly study physical properties of models with the temperature given by a two-periodic sequence.

\subsection*{Organization of the paper}

Section~\ref{sec:OGpolynomials} is devoted to a general study of a new family of symmetric quantum orthogonal polynomials. We deduce formulas for their asymptotic behavior, describe their measure of  orthogonality, and derive generating functions. An important role is played by $q$-hypergeometric series, the $q$-Gauss hypergeometric function, in the first place.

The application to the Ising model indicated above is treated in Section~\ref{sec:Ising} in detail. The main results for the magnetization and the two-spin correlation for one-dimensional semi-infinite kinetic Ising chain with hyperbolically decreasing temperature are summarized in Theorems~\ref{thm:mag} and~\ref{thm:two-spin_cor}, respectively.

For the reader's convenience, the text is accompanied by two appendices containing a brief overview of $q$-hypergeometric functions, selected formulas used within the paper, and properties of zeros of the  Gauss $q$-hypergeometric function that are needed in Section~\ref{sec:OGpolynomials}.

%
%
\section{A family of symmetric quantum orthogonal polynomials}
\label{sec:OGpolynomials}

Consider a family of monic orthogonal polynomials given by the tree-term recurrence
\begin{equation}
  \label{eq:three-term_recur_gener}
  p_{n+1}(x)=(x-b_{n})p_{n}(x)-a_{n-1}p_{n-1}(x), \quad n\in\N_{0},
\end{equation}
with initial conditions $p_{-1}(x)=0$ and $p_{0}(x)=1$, where $b_{n}\in\R$ and $a_{n}>0$ for all $n\in\N_{0}$.
The standard reference books on the theory of orthogonal polynomials are~\cite{chihara78, ismail09,nevai_mams79,szego75}.

Recall that if the Jacobi parameters from the three-term recurrence~\eqref{eq:three-term_recur_gener} satisfy
\begin{equation}
  \label{eq:nevai_cond_recur_coef}
  \sum_{n=0}^{\infty}n\Big(|b_{n}| + \big|1-\sqrt{a_{n}}\big|\Big) < \infty,
\end{equation}
then the measure of orthogonality $\mu$ of polynomials $p_{n}$ decomposes as
\begin{equation}
  \label{eq:meas_decomp_ac_d}
  \mu=\mu_{ac}+\mu_{d},
\end{equation}
where $\mu_{ac}$ is absolutely continuous w.r.t. the Lebesgue measure, $\supp\mu_{ac} = [-2,2]$, and $\mu_{d}$ is a discrete measure supported on 
finitely many points located in $\R\setminus[-2,2]$; see \cite{vanassche_90}.
Hence we may write
\[
  \frac{\dd\mu_{ac}}{\dd x}(x)=\rho(x) \quad \mbox{ and } \quad \mu_{d}(x)=\sum_{k=1}^{N}w_{k}\delta_{x_{k}}(x),
\]
where $\rho$ is, in fact, a positive function on $(-2,2)$, $N\in\N_{0}$, $w_{k}>0$ and $|x_{k}|>2$ for all $1\leq k\leq N$.
The orthogonality relation for polynomials $p_{n}$ then reads
\begin{equation}
  \label{eq:og_rel_general}
  \int_{-2}^{2}p_{n}(x)p_{m}(x)\rho(x)\dd x+\sum_{k=1}^{N}w_{k}p_{m}(x_{k})p_{n}(x_{k})=\left(\prod_{j=0}^{n-1}a_{j}\right)\!\delta_{m,n},
\end{equation}
for $m,n\in\N_{0}$.

In the following, we always assume that $q\in(0,1)$.
Moreover, in the notation, we will follow Gasper and Rahman's book \cite{gasper90} on the basic hypergeometric (or $q$-hypergeometric) series. 
For reader's convenience, basic definitions and some formulas with $q$-hypergeometric series relevant to the text are summarized in Appendix~\ref{app:A}.

Our primary goal is to investigate polynomials
$p_{n}(x) \equiv p_{n}^{(\alpha,\beta)}(x;q)$ 
generated by recurrence relation
\begin{equation}
  \label{eq:ogp_monic_recur}
  p_{n+1}(x)=xp_{n}(x)-\tilde\gamma_{n-1}\tilde\gamma_{n}p_{n-1}(x), \quad n\in\N_{0},
\end{equation}
and initial conditions $p_{-1}(x)=0$ and $p_{0}(x)=1$, where 
\begin{equation}
  \label{eq:def_gamma_n_q_form}
  \tilde\gamma_{n}=\tilde\gamma_{n}^{(\alpha,\beta)}(q):=\frac{1-\alpha q^{n}}{1-\beta q^{n}}, \quad n\in\N_{0}.
\end{equation}
At the moment, the parameters $\alpha$ and $\beta$ are assumed to be real and smaller than $1$ but later on they will be restricted even further.
Hence, polynomials $p_{n}^{(\alpha,\beta)}(x;q)$ satisfy~\eqref{eq:three-term_recur_gener} with
\[
  b_{n}=0 \quad \text{and} \quad a_{n} = \tilde\gamma_{n} \tilde\gamma_{n+1}, \quad n\in\N_{0}.
\]
In this case, condition~\eqref{eq:nevai_cond_recur_coef} is clearly fulfilled and therefore the measure of orthogonality decomposes as in~\eqref{eq:meas_decomp_ac_d}. 
Moreover, since $b_{n}=0$ for all $n\in\N_{0}$, the orthogonal polynomials $p^{(\alpha,\beta)}(x;q)$ are symmetric, i.e, their measure of orthogonality $\mu$ is symmetric with respect to the origin which means $\mu(\mathcal{A})=\mu(-\mathcal{A})$ for any Borel set $\mathcal{A}\subset\R$; see \cite{chihara78}.

Let us denote by $\{q_{n}\}_{n=0}^{\infty}$ the sequence of monic orthogonal polynomials of the second kind which is the solution of the recurrence~\eqref{eq:three-term_recur_gener}
determined by the initial conditions $q_{0}(x)=0$ and $q_{1}(x)=1$. If $q_{n}(x) \equiv q_{n}^{(\alpha,\beta)}(x;q)$ stands for the orthogonal polynomials of the second kind corresponding
to $p_{n}^{(\alpha,\beta)}(x;q)$, then one easily verifies that
\begin{equation}
  \label{eq:ogp_sec_kind} 
  q_{n}^{(\alpha,\beta)}(x;q)=p_{n-1}^{(\alpha q,\beta q)}(x;q), \quad n\in\N_{0},
\end{equation}
using the identity $\tilde\gamma_{n}^{(\alpha,\beta)}(q) = \tilde\gamma_{n-1}^{(\alpha q,\beta q)}(q)$.

Recall further that the Cauchy transform of the measure $\mu$ is defined by
\[
  C_{\mu}(z):=\int_{\mathbb{R}}\frac{\dd\mu(x)}{z-x}, \quad z\in\C\setminus\supp\mu.
\]
According to the Markov Theorem, 
\begin{equation}
  \label{eq:markov}
  C_{\mu}(z)=\lim_{n\to\infty}\frac{q_{n}(z)}{p_{n}(z)}, \quad z\in\C\setminus\supp\mu,
\end{equation}
provided that the sequences $\{a_{n}\}_{n=0}^{\infty}$ and $\{b_{n}\}_{n=0}^{\infty}$ are bounded, see, for example,~\cite{berg_jat94}. 
Moreover, the sequence on the right-hand side of~\eqref{eq:markov} converges uniformly on any compact subset of $\C\setminus\supp\mu$.

\begin{rem}
  \label{rem:OriginalIsing}
  The particular case with $q=e^{-2\kappa}$, where $\kappa>0$, and $\alpha=-\beta=q$ yields
  \[
  \tilde\gamma_{n} = \tanh\left((n+1)\kappa\right) = \gamma_{n+1}, \quad n\in\N_{0},
  \]
  which brings us to the initial motivation of the solvable kinetic Ising model, cf. \eqref{eq:gamma_tanh}. 
\end{rem}

\subsection{Solutions of a second-order difference equation and asymptotic formulas for the orthogonal polynomials}

Recall the Joukowsky conformal map 
\begin{equation}
\upsilon(z):=z+z^{-1}
\label{eq:joukowsky}
\end{equation}
maps the punctured unit disk $\{z\in\C \mid 0<|z|<1\}$
bijectively onto $\C\setminus[-2, 2]$ and the half-circle $\{e^{\ii\theta}\mid \theta\in[0,\pi]\}$ 
bijectively onto $[-2, 2]$. To deduce the desired properties of the orthogonal polynomials $p_{n}^{(\alpha,\beta)}$,
we first investigate the solutions of the second-order difference equation
\begin{equation}
 \psi_{n-1}-\upsilon(z)\frac{1-\beta q^{n}}{1-\alpha q^{n}}\psi_{n}+\psi_{n+1}=0, \quad n\in\N_{0},
\label{eq:diff_eq}
\end{equation}
where $\alpha<1$ and $\beta<1$. Note that, if a sequence $\psi$ is a solution of \eqref{eq:diff_eq},
then $u$, where
\[
 u_{n}:=\frac{(\alpha;q)_{n}}{(\beta;q)_{n}}\psi_{n}, \quad n\in\N_{0},
\]
solves \eqref{eq:ogp_monic_recur} with $x=\upsilon(z)$. For the definitions of the $q$-Pochhhamer symbols and other basic definitions from the theory of $q$-hypergeometric series, see Appendix~\ref{app:A}.

For $|\alpha|<1$ and $z\neq0$, we define
\begin{equation}
 \psi_{n}^{\pm}(z):=z^{\pm n}\left(qz^{\pm2};q\right)_{\!\infty}\pFq{2}{1}{z^{\pm1}\tau,z^{\pm1}\tau^{-1}}{qz^{\pm2}}{\alpha q^{n+1}}, \quad n\geq-1,
 \label{eq:def_psi_pm}
\end{equation}
where the variable $\tau$ is determined by the equation
\begin{equation}
\alpha\upsilon(\tau)=\beta\upsilon(z),
\label{eq:z_to_tau_rel}
\end{equation}
$\upsilon$ is the Joukowsky map~\eqref{eq:joukowsky}, and ${}_2\phi_1$ is the $q$-Gauss hypergeometric series.

\begin{rem}
Note that, if $|\alpha|\geq1$, the $q$-Gauss hypergeometric series in~\eqref{eq:def_psi_pm} diverges for $n=-1$.
Further, observe that $\psi_{n}^{\pm}(z)$ are analytic in $z\in\C\setminus\{0\}$. This is the reason why the 
$n$-independent factor $\left(qz^{\pm2};q\right)_{\!\infty}$ is present in~\eqref{eq:def_psi_pm}.
\label{rem:basic_prop_psi_pm}
\end{rem}

\begin{rem}
If $\alpha=0$ or $\beta=0$, $\psi_{n}^{\pm}$ is defined by the respective limit values, i.e.,
\[
 \psi_{n}^{\pm}(z):=z^{\pm n}\left(qz^{\pm2};q\right)_{\!\infty}\pFq{2}{1}{\ii z^{\pm1},-\ii z^{\pm1}}{qz^{\pm2}}{\alpha q^{n+1}}, \quad \mbox{ if } \beta=0,
\]
and
\[
 \psi_{n}^{\pm}(z):=z^{\pm n}\left(qz^{\pm2};q\right)_{\!\infty}\pFq{1}{1}{0}{qz^{\pm2}}{\beta(1+z^{2})q^{n+1}}, \quad \mbox{ if } \alpha=0.
\]
\end{rem}

The following proposition shows that the sequences $\psi^{\pm}(z)$ form generically a couple of linearly independent solutions of the second-order difference equation~\eqref{eq:diff_eq}.

\begin{prop}\label{prop:two_sol_psi_pm}
 For any $\alpha,\beta,z\in\C$ such that $|\alpha|<1$ and $z\neq0$, the sequences $\psi^{\pm}(z)$ given by~\eqref{eq:def_psi_pm} are two solutions of~\eqref{eq:diff_eq}. 
 Moreover, for their Wronskian 
 $W(\psi^{+},\psi^{-}):=\psi_{n}^{+}(z)\psi_{n+1}^{-}(z)-\psi_{n+1}^{+}(z)\psi_{n}^{-}(z)$, one has 
 \begin{equation}
  W(\psi^{+},\psi^{-})=z^{-1}\left(z^{2},qz^{-2};q\right)_{\!\infty}. 
 \label{eq:wrons_psipm}
 \end{equation}
 Consequently, $\psi^{\pm}(z)$ are linearly independent if and only if $z\notin \pm q^{\mathbb{Z}/2}$.
\end{prop}
%
%

\begin{proof}
 First, we verify that $\psi^{+}$ solves~\eqref{eq:diff_eq}. Let us consider the $q$-difference equation
 \begin{equation}
  (1-\alpha\xi)\left(f(q\xi)+f(q^{-1}\xi)\right)-\upsilon(z)(1-\beta\xi)f(\xi)=0
 \label{eq:diff_eq_in_proof}
 \end{equation}
 and look for its solution in the form
 \[
  f(\xi)=\sum_{n=0}^{\infty}a_{n}\xi^{n+\log_{q}z}.
 \]
 By plugging this Ansatz into~\eqref{eq:diff_eq_in_proof}, one obtains the first-order difference equation for the coefficients
 \[
  \left(zq^{n}+z^{-1}q^{-n}-z-z^{-1}\right)a_{n}=\alpha\left(zq^{n-1}+z^{-1}q^{-n+1}-\tau-\tau^{-1}\right)a_{n-1},
 \]
 where the new variable $\tau$ is determined by the equation~\eqref{eq:z_to_tau_rel}. The last equation can be rewritten as
 \[
  a_{n}=q\alpha\frac{\left(1-z\tau q^{n-1}\right)\left(1-z\tau^{-1} q^{n-1}\right)}{\left(1-z^{2} q^{n}\right)\left(1-q^{n}\right)}a_{n-1}, \quad n\in\N.
 \]
 Hence the solution reads
 \[
  a_{n}=\left(q\alpha\right)^{n}\frac{\left(z\tau,z\tau^{-1};q\right)_{n}}{\left(q,qz^{2};q\right)_{n}}a_{0}, \quad n\in\N_{0}.
 \]
 Consequently, the function
 \[
  f(\xi)=\xi^{\log_{q} z}\pFq{2}{1}{z\tau,z\tau^{-1}}{qz^{2}}{\alpha q\xi}
 \]
 is a solution of~\eqref{eq:diff_eq_in_proof} for any $z\notin\{0\}\cup q^{-\N/2}$, $\alpha,\beta\in\C\setminus\{0\}$, and $\xi\in\C$ such that $|\alpha q\xi|<1$.
 
 It follows readily from the relation
 \[
  \psi_{n}^{+}(z)=f\left(q^{n}\right)\!,
 \]
 that $\psi^{+}$ is a solution of~\eqref{eq:diff_eq}. In addition, since the equation~\eqref{eq:diff_eq_in_proof} 
 is invariant under the exchange of $z$ and $z^{-1}$ and $\psi^{-}(z)=\psi^{+}(z^{-1})$, $\psi^{-}$ is a solution 
 of~\eqref{eq:diff_eq}, too. Further, since $\psi^{(\pm)}$ are analytic in $z\in\C\setminus\{0\}$, they have to be 
 solutions of~\eqref{eq:diff_eq} for all $z\in\C\setminus\{0\}$. Similarly, the validity of the statement can be 
 extended for the values $\alpha=0$ or $\beta=0$ by the continuity.
 
 Second, recall that the Wronskian
 \begin{equation}
  W(\psi^{+},\psi^{-})=\psi_{n}^{+}(z)\psi_{n+1}^{-}(z)-\psi_{n+1}^{+}(z)\psi_{n}^{-}(z)
 \label{eq:def_wronsk_general}
 \end{equation}
 is a constant not depending on $n$ that vanishes if and only if $\psi^{\pm}(z)$ are linearly dependent solutions.
 By using the definition of the $q$-Gauss hypergeometric series in~\eqref{eq:def_psi_pm}, one immediately obtains
 the asymptotic formulas
 \begin{equation}
  \psi_{n}^{\pm}(z)=z^{\pm n}\left(qz^{\pm2};q\right)_{\!\infty}\left(1+O\left(q^{n}\right)\right)\!, \quad \mbox{ for } n\to\infty.
 \label{eq:asympt_psipm_n_infpos}
 \end{equation}
 By sending $n\to\infty$ on the right-hand side of~\eqref{eq:def_wronsk_general}, one obtains the identity~\eqref{eq:wrons_psipm}. Clearly, the right-hand side of~\eqref{eq:wrons_psipm}
 vanishes if and only if $z\in \pm q^{\mathbb{Z}/2}$.
\end{proof}

Knowing the two solutions of~\eqref{eq:diff_eq}, one can deduce a formula for the orthogonal polynomials~$p_{n}$ 
in terms of the $q$-Gauss hypergeometric series.

\begin{prop}
 For $|\alpha|<1$, $\beta<1$, $z\neq0$, and $n\geq-1$, one has
 \begin{align}
  p_{n}^{(\alpha,\beta)}(\upsilon(z);q)=\frac{1}{z^{-1}-z}\frac{(\alpha;q)_{n}}{(\beta;q)_{n}}
  &\bigg[z^{-n-1}\pFq{2}{1}{z\tau,z\tau^{-1}}{qz^{2}}{\alpha }\pFq{2}{1}{z^{-1}\tau,z^{-1}\tau^{-1}}{qz^{-2}}{\alpha q^{n+1} }\nonumber\\
  &-z^{n+1}\pFq{2}{1}{z^{-1}\tau,z^{-1}\tau^{-1}}{qz^{-2}}{\alpha }\pFq{2}{1}{z\tau,z\tau^{-1}}{qz^{2}}{\alpha q^{n+1} } \bigg]\!.\nonumber\\
  \label{eq:ogp_formula_psipm}
 \end{align}
 If $z\in \pm q^{\Z/2}$, then the right-hand side has to be understood as the respective limit value.
\end{prop}
%
%

\begin{rem}
 Note that it is by no means obvious that the right-hand side of~\eqref{eq:ogp_formula_psipm} is a polynomial in~$\upsilon(z)$.
\end{rem}

\begin{proof}
 Let $z\notin \pm q^{\Z/2}$. The sequence
 \[
  u_{n}(z):=\frac{(\alpha;q)_{n}}{(\beta;q)_{n}}\psi_{n}(z),
 \]
 where
 \[
 \psi_{n}(z):=\frac{\psi_{n}^{-}(z)\psi_{-1}^{+}(z)-\psi_{n}^{+}(z)\psi_{-1}^{-}(z)}{\psi_{0}^{-}(z)\psi_{-1}^{+}(z)-\psi_{0}^{+}(z)\psi_{-1}^{-}(z)}=
 \frac{\psi_{n}^{-}(z)\psi_{-1}^{+}(z)-\psi_{n}^{+}(z)\psi_{-1}^{-}(z)}{W(\psi^{+},\psi^{-})},
 \]
 is the solution of~\eqref{eq:ogp_monic_recur}, with $x$ replaced by $\upsilon(z)$, satisfying the same initial 
 conditions as $p_{n}$, i.e., $u_{-1}(z)=0$ and $u_{0}(z)=1$. Thus, $p_{n}^{(\alpha,\beta)}(\upsilon(z);q)=u_{n}(z)$ 
 for all $n\geq-1$. To arrive at the formula~\eqref{eq:ogp_formula_psipm}, it suffices to use the formula for the 
 Wronskian~\eqref{eq:wrons_psipm} and the definition~\eqref{eq:def_psi_pm}. The resulting formula can be extended 
 to all $z\neq0$ by continuity.
\end{proof}

An advantage of the rather complicated expression~\eqref{eq:ogp_formula_psipm} for the polynomial $p_{n}^{(\alpha,\beta)}$ 
is that one can immediately derive the asymptotic behavior of $p_{n}$ for $n\to\infty$.

\begin{cor}
 Let $|\alpha|<1$, $\beta<1$, then, as $n\to\infty$, one has the asymptotic expansion
 \begin{align*}
  p_{n}^{(\alpha,\beta)}(\upsilon(z);q)=\frac{1}{z^{-1}-z}\frac{(\alpha;q)_{\infty}}{(\beta;q)_{\infty}}&
  \bigg(z^{-n-1}\pFq{2}{1}{z\tau,z\tau^{-1}}{qz^{2}}{\alpha } \\
  &\hskip18pt-z^{n+1}\pFq{2}{1}{z^{-1}\tau,z^{-1}\tau^{-1}}{qz^{-2}}{\alpha }\bigg)\left(1+O(q^{n})\right).
 \end{align*}
 In particular, if $0<|z|<1$, then
 \begin{equation}
  p_{n}^{(\alpha,\beta)}(\upsilon(z);q)=\frac{1}{z^{-1}-z}\frac{(\alpha;q)_{\infty}}{(\beta;q)_{\infty}}
  z^{-n-1}\pFq{2}{1}{z\tau,z\tau^{-1}}{qz^{2}}{\alpha }\left(1+O\left(\left[\max(|z|,q)\right]^{n}\right)\right),
 \label{eq:p_n_asympt_inside_unit disk}
 \end{equation}
 and, if $\theta\in(0,\pi)$, then
 \begin{equation}
  p_{n}^{(\alpha,\beta)}(2\cos\theta;q)=\frac{1}{\sin\theta}\frac{(\alpha;q)_{\infty}}{(\beta;q)_{\infty}}
  \Im\left[e^{\ii(n+1)\theta}\pFq{2}{1}{e^{-\ii\theta}\tau,e^{-\ii\theta}\tau^{-1}}{qe^{-2\ii\theta}}{\alpha }\right]\left(1+O\left(q^{n}\right)\right).
  \label{eq:p_n_asympt_essen}
 \end{equation}
\end{cor}
%
%
%
%
%
%
%
%

\subsection{The measure of orthogonality} 

In this subsection, a detailed description of the absolutely continuous and the discrete part of the measure of orthogonality
for polynomials $p_{n}^{(\alpha,\beta)}$ is derived. These goals are achieved by means of the Cauchy transform
that can be expressed as a ratio of two $q$-Gauss hypergeometric functions.

\begin{prop}\label{prop:Cauchy_transf_mu}
 For $|\alpha|<1$, $\beta<1$, and $0<|z|<1$, one has
 \[
  C_{\mu}(\upsilon(z))=\frac{1-\beta}{1-\alpha}\frac{\psi_{0}^{+}(z)}{\psi_{-1}^{+}(z)}=
  z\frac{1-\beta}{1-\alpha}\,\pFq{2}{1}{z\tau,z\tau^{-1}}{qz^{2}}{\alpha q}\bigg/\pFq{2}{1}{z\tau,z\tau^{-1}}{qz^{2}}{\alpha},
 \]
 where $\tau$ is determined by the equation~\eqref{eq:z_to_tau_rel}.
\end{prop}

\begin{proof}
 It follows immediately from the Markov Theorem~\eqref{eq:markov}, equality~\eqref{eq:ogp_sec_kind}, and asymptotic formula~\eqref{eq:p_n_asympt_inside_unit disk}.
\end{proof}

First, we derive the density of the absolutely continuous component of $\mu$.

\begin{prop}\label{prop:ac_meas}
 For $|\alpha|<1$, $\beta<1$ and $\theta\in(0,\pi)$, one has
 \[
  \frac{\dd\mu_{ac}}{\dd\theta}(2\cos\theta)=-\frac{1-\beta}{2\pi(1-\alpha)}\left|\frac{\left(e^{2\ii\theta};q\right)_{\!\infty}}{\psi_{-1}^{+}(e^{\ii\theta})}\right|^{2}
  =-\frac{2(1-\beta)}{\pi(1-\alpha)}\sin^{2}(\theta)\left|\pFq{2}{1}{e^{\ii\theta}\tau,e^{\ii\theta}\tau^{-1}}{qe^{2\ii\theta}}{\alpha}\right|^{-2},
 \]
 where $2\beta\cos\theta=\alpha(\tau+\tau^{-1})$.
\end{prop}

\begin{proof}
 The measure $\mu$ can be recovered from its Cauchy transform $C_{\mu}$ by using the Stieltjes--Perron inversion formula. In particular, for the absolutely continuous part of $\mu$, it takes the form
 \[
  \frac{\dd\mu_{ac}}{\dd x}(x)=-\frac{1}{\pi}\lim_{\substack{u\to x \\ \Im u>0}}\Im C_{\mu}(u), \quad x\in(-2,2).
 \]
 We put $u=\upsilon(z)$. Note that, for some $\theta\in(0,\pi)$, $u$ approaches the point $2\cos\theta$ from the upper half-plane $\Im u>0$ if and only if $z$ approaches the point $e^{-\ii\theta}$
 from the interior of the unit disk $|z|<1$. Consequently, by using Proposition~\ref{prop:Cauchy_transf_mu}, one gets
 \[
  \lim_{\substack{u\to x \\ \Im u>0}}\Im C_{\mu}(u)=\lim_{\substack{z\to e^{-\ii\theta} \\ |z|<1}}\Im C_{\mu}(\upsilon(z))=
  \frac{1-\beta}{1-\alpha}\Im\left(\frac{\psi_{0}^{+}(e^{-\ii\theta})}{\psi_{-1}^{+}(e^{-\ii\theta})}\right)\!.
 \]
 Further, since 
 \[
 \overline{\psi_{n}^{+}(e^{-\ii\theta})}=\psi_{n}^{-}(e^{-\ii\theta}), \quad \forall n\geq-1,
 \]
 one has
 \[
  \Im\left(\frac{\psi_{0}^{+}(e^{-\ii\theta})}{\psi_{-1}^{+}(e^{-\ii\theta})}\right)=
  \frac{\psi_{0}^{+}(e^{-\ii\theta})\psi_{-1}^{-}(e^{-\ii\theta})-\psi_{-1}^{+}(e^{-\ii\theta})\psi_{0}^{-}(e^{-\ii\theta})}{2\ii\left|\psi_{-1}^{+}(e^{-\ii\theta})\right|^{2}}.
 \]
 One recognizes the Wronskian in the nominator above. Thus, recalling~\eqref{eq:wrons_psipm}, one obtains
 \[
  \Im\left(\frac{\psi_{0}^{+}(e^{-\ii\theta})}{\psi_{-1}^{+}(e^{-\ii\theta})}\right)=
  -\frac{e^{\ii\theta}\left(e^{-2\ii\theta},qe^{2\ii\theta};q\right)_{\infty}}{2\ii\left|\psi_{-1}^{+}(e^{-\ii\theta})\right|^{2}}
  =-\frac{1}{4\sin\theta}\left|\frac{\left(e^{2\ii\theta};q\right)_{\infty}}{\psi_{-1}^{+}(e^{\ii\theta})}\right|^{2}\!.
 \]
 
 In total, we have
 \[
  \frac{\dd\mu_{ac}}{\dd x}(2\cos\theta)=\frac{1-\beta}{4\pi(1-\alpha)\sin\theta}\left|\frac{\left(e^{2\ii\theta};q\right)_{\infty}}{\psi_{-1}^{+}(e^{\ii\theta})}\right|^{2}
 \]
 and hence
 \[
  \frac{\dd\mu_{ac}}{\dd \theta}(2\cos\theta)=-\frac{1-\beta}{2\pi(1-\alpha)}\left|\frac{\left(e^{2\ii\theta};q\right)_{\infty}}{\psi_{-1}^{+}(e^{\ii\theta})}\right|^{2}, \quad  \mbox{ for } \theta\in(0,\pi).
 \]
 The second expression of the density from the statement follows from the definition~\eqref{eq:def_psi_pm} for the function $\psi_{-1}^{+}$.
\end{proof}

\begin{rem}
 Note that the function on the right-hand side for the density $\dd\mu_{ac}/\dd \theta$ from Proposition~\ref{prop:ac_meas} remains unchanged if $\theta$ 
 is replaced by $\pi-\theta$ as expected.
\end{rem}

\begin{rem}
 An alternative proof of Proposition~\ref{prop:ac_meas} can be based on the asymptotic formula~\eqref{eq:p_n_asympt_essen} 
 and the result of Nevai, see~\cite[Cor.~36]{nevai_mams79}, which implies that
 \[
  \limsup_{n\to\infty}\frac{\dd\mu_{ac}}{\dd\theta}(2\cos\theta)\frac{(q\beta,\beta;q)_{n-1}}{(q\alpha,\alpha;q)_{n-1}}
  \left(p_{n}^{(\alpha,\beta)}(2\cos\theta)\right)^{2}=-\frac{2}{\pi},
 \]
 for a.~e. $\theta\in(0,\pi)$.
\end{rem}

For the purpose of a description of the discrete part of the measure $\mu$, let us denote by $z_{k}=z_{k}(\alpha,\beta;q)$, 
$k\in\N$, the positive zeros of $\psi_{-1}^{+}$ ordered increasingly, i.e., $0<z_{1}<z_{2}<\dots$ A certain analysis 
of the zeros of $\psi_{n}^{+}$ is needed at this point. In order to remain focused on the description of the measure $\mu$
in this subsection, the analysis of zeros is postponed to Appendix~\ref{app:B}.

\begin{prop}\label{prop:d_meas}
 Let $|\alpha|<1$, $\beta<1$. If $z_{1}\geq1$, then $\mu_{d}=0$. Otherwise, $\mu_{d}$ is supported on finitely many 
 points $\pm\upsilon(z_{k})$, where $z_{k}$ are the positive zeros of $\psi_{-1}^{+}$ located in $(0,1)$, 
 and the the measure takes the form
 \[
  \mu_{d}=\frac{1-\beta}{1-\alpha}\sum_{0<z_{k}<1}\frac{z_{k}^{2}-1}{z_{k}^{2}}\frac{\psi_{0}^{+}(z_{k})}{\left(\psi_{-1}^{+}\right)'(z_{k})}\left(\delta_{\upsilon(z_{k})}+\delta_{-\upsilon(z_{k})}\right).
 \]
\end{prop}

\begin{proof}
 The discrete part of the measure $\mu$ is supported on the set of poles of its Cauchy transform $C_{\mu}$. By using Proposition~\ref{prop:Cauchy_transf_mu} and Proposition~\ref{prop:zeros_basic_propert},
 one observes that $\mu_{d}$ is supported on the set of points $\pm\upsilon(z_{k})$ where $z_{k}$ are the zeros of $\psi_{-1}^{+}$ located in $(0,1)$. If $z_{1}\geq1$, this set is empty and hence $\mu_{d}=0$.
 
 If $z_{1}<1$, $\mu_{d}$ is a linear combination of normalized Dirac measures supported on $\pm\upsilon(z_{k})$, $z_{k}\in(0,1)$, i.e.,
 \[
  \mu_{d}=\sum_{0<z_{k}<1}w_{k}\delta_{\upsilon(z_{k})}+\tilde{w}_{k}\delta_{-\upsilon(z_{k})}.
 \]
 Since the measure $\mu$ is symmetric with respect to the origin, $w_{k}=\tilde{w}_{k}$.
 Further, since $\pm\upsilon(z_{k})$ are isolated points of $\supp\mu$, the weight $w_{k}$ can be computed from $C_{\mu}$ as
 \[
  w_{k}=\frac{1}{2\pi\ii}\oint_{|u-\upsilon(z_{k})|=\epsilon}C_{\mu}(u)\dd u=\frac{1}{2\pi\ii}\oint_{|z-z_{k}|=\epsilon}C_{\mu}(\upsilon(z))\left(1-z^{-2}\right)\dd z,
 \]
 with $\epsilon>0$ sufficiently small. It follows from Proposition~\ref{prop:Cauchy_transf_mu} and Proposition~\ref{prop:zeros_basic_propert} that $z_{k}$ is a simple 
 pole of $C_{\mu}$ and hence, by the Residue Theorem, one has
 \[
   w_{k}=\Res_{z=z_{k}}\left(\frac{z^{2}-1}{z^{2}}C_{\mu}(\upsilon(z))\right)=\frac{1-\beta}{1-\alpha}\frac{z_{k}^{2}-1}{z_{k}^{2}}\frac{\psi_{0}^{+}(z_{k})}{\left(\psi_{-1}^{+}\right)'(z_{k})}. \qedhere
 \]
\end{proof}

Proposition~\ref{prop:d_meas} shows that the description of the discrete part of the measure $\mu$ heavily depends on the location of zeros of the function 
$\psi_{-1}^{+}$. Therefore a more detailed analysis of the zeros would be desirable. For example, if $N(\alpha,\beta;q)$ denotes the number of points in $\{z_{k}\}_{k=1}^{\infty}\cap(0,1)$,
it would be interesting to know more about properties of $N(\alpha,\beta;q)$ as function of $\alpha$, $\beta$, and $q$.
Nevertheless, we are able to prove only a simple sufficient condition guaranteeing that $N(\alpha,\beta;q)=0$ at this point.

\begin{prop}\label{prop:mu_d_vanish_suf_cond}
 Let $|\alpha|<1$ and $\beta<1$. If $\beta\leq\alpha$, then $\mu_{d}=0$.
\end{prop}

\begin{proof}
 Note that, if $\beta\leq\alpha$, then $|\tilde\gamma_{n}| \leq 1$, $\forall n\in\N_{0}$, where $\tilde\gamma_{n}$ is defined by~\eqref{eq:def_gamma_n_q_form}.
 Thus, for the norm of the Jacobi operator $J$ whose diagonal vanishes and the off-diagonal is given by the sequence $\{\tilde\gamma_{n}\tilde\gamma_{n+1}\}_{n=0}^{\infty}$, one has
 \[
  \|J\|\leq2\sup_{n\in\N_{0}}|\tilde\gamma_{n}\tilde\gamma_{n+1}|\leq 2.
 \]
 Consequently, it holds that 
 \[
  \supp\mu=\spec(J)\subset[-2,2].
 \]
 Since the discrete part of $\mu$, if present, has to be supported in $\R\setminus[-2,2]$, $\mu_{d}$ vanishes.
 \end{proof}

As a conclusion, we use both formulas for the absolutely continuous and the discrete part of $\mu$ 
obtained in Propositions~\ref{prop:ac_meas} and~\ref{prop:d_meas} to state the orthogonality relation for the polynomials $p_{n}=p_{n}^{(\alpha,\beta)}$.

\begin{thm}
 Let $|\alpha|<1$, $\beta<1$. Then, for all $m,n\in\N_{0}$, one has the orthogonality relation
 \begin{align*}
  \frac{2}{\pi}&\int_{0}^{\pi}p_{m}(2\cos\theta)p_{n}(2\cos\theta)\sin^{2}(\theta)\left|\pFq{2}{1}{e^{\ii\theta}\tau,e^{\ii\theta}\tau^{-1}}{qe^{2\ii\theta}}{\alpha}\right|^{-2}\dd\theta\\
  &+\sum_{0<z_{k}<1}\frac{z_{k}^{2}-1}{z_{k}^{2}}\frac{\psi_{0}^{+}(z_{k})}{\left(\psi_{-1}^{+}\right)'(z_{k})}\left[p_{m}(\upsilon(z_{k}))p_{n}(\upsilon(z_{k}))+p_{m}(-\upsilon(z_{k}))p_{n}(-\upsilon(z_{k}))\right]\\
  &\hskip260pt=\frac{(\alpha;q)_{n}(\alpha;q)_{n+1}}{(\beta;q)_{n}(\beta;q)_{n+1}}\delta_{m,n},
 \end{align*}
 where $z_{k}$, $k\in\N$, stands for the $k$th zero of $\psi_{-1}^{+}$ in the interval $(0,1)$.
\end{thm}
%
%

\begin{rem}
 Note that the variable $\tau$ present in the integral above is related to~$\theta$ via the equation $\alpha(\tau+\tau^{-1})=2\beta\cos\theta$.
 While the $\tau$ hidden in the definition of $\psi^{+}_{n}$, $n\in\{-1,0\}$, that appear in the sum above, depends on $k$
 according to the equation $\alpha(\tau+\tau^{-1})=\beta(z_{k}+z_{k}^{-1})$. Recall also that the discrete part of the measure $\mu$
 can vanish and hence the sum can be empty.
\end{rem}

\subsection{The generating function} 

Next, we derive a generating function formula for $p_{n}^{(\alpha,\beta)}$ 
provided that the range of the parameter $\alpha$ is restricted to $\alpha\in(-q,q]$.

\begin{thm}
 For $|t|<|z|<1$, one has:
 \begin{enumerate}
  \item if $|\alpha|<q$, then
  \begin{equation}
   \sum_{n=0}^{\infty}\frac{(\beta;q)_{n}}{(\alpha;q)_{n}}p_{n}^{(\alpha,\beta)}(\upsilon(z);q)t^{n}
   =\frac{1-q^{-1}\alpha}{(1-zt)(1-z^{-1}t)}\,\pFq{3}{2}{q,q\tau t,q\tau^{-1}t}{qzt,qz^{-1}t}{q^{-1}\alpha};
   \label{eq:gener_func_alp<q}
  \end{equation}
  \item for $\alpha=q$, it holds
  \begin{equation}
   \sum_{n=0}^{\infty}\frac{(\beta;q)_{n}}{(q;q)_{n}}p_{n}^{(q,\beta)}(\upsilon(z);q)t^{n}=\frac{(q\tau t,q\tau^{-1}t;q)_{\infty}}{(zt,z^{-1}t;q)_{\infty}}.
   \label{eq:gener_func_alp=q}
  \end{equation}
 \end{enumerate}
 The variable $\tau$ is determined by the equation~\eqref{eq:z_to_tau_rel}.
\end{thm}

\begin{rem}
 Recall that in the analysis of properties of polynomials $p_{n}^{(\alpha,\beta)}$ we usually assume $|\alpha|<1$. 
 A (convergent) generating function formula for $p_{n}^{(\alpha,\beta)}$ for $\alpha\notin(-q,q]$ remains unknown at the
 moment.
\end{rem}

\begin{proof}
 1) Let $0<|z|<1$ be fixed. Let $U(t)$ denote the right-hand side of~\eqref{eq:gener_func_alp<q}. The function $U$ is analytic in the disk $|t|<|z|$.
 It is easy to verify that $U$ satisfies the first order $q$-difference equation
 \begin{equation}
  (1-zt)(1-z^{-1}t)U(t)=1-q^{-1}\alpha+q^{-1}\alpha(1-q\tau t)(1-q\tau^{-1}t)U(qt).
 \label{eq:gener_q_dif}
 \end{equation}
 If we write the Taylor expansion of $U$ in the form
 \[
  U(t)=\sum_{n=0}^{\infty}\frac{(\beta;q)_{n}}{(\alpha;q)_{n}}c_{n}t^{n}
 \]
 and plug it into~\eqref{eq:gener_q_dif}, then one obtains the recurrences
 \[
  c_{1}=\upsilon(z)c_{0} \quad \mbox{ and } \quad c_{n+1}=\upsilon(z)c_{n}-\frac{(1-\alpha q^{n})(1-\alpha q^{n-1})}{(1-\beta q^{n})(1-\beta q^{n-1})}c_{n-1}, \quad \mbox{ for } n\in\N.
 \]
 Noticing also that $c_{0}=1$, one concludes that $c_{n}=p_{n}^{(\alpha,\beta)}(\upsilon(z);q)$ for all $n\in\N_{0}$.
 
 2) To prove \eqref{eq:gener_func_alp<q}, one proceeds similarly. If, this time, $U(t)$ stands for the right-hand side of~\eqref{eq:gener_func_alp=q}, then $U$ is analytic on $|t|<|z|$ and still
 satisfies~\eqref{eq:gener_q_dif} with $\alpha=q$. The rest is analogous.
\end{proof}

\subsection{Special case \texorpdfstring{$\alpha=q$}{alpha = q}}
\label{subsec:specialCase}

The obtained formulas get a very explicit form in the particular case when $\alpha=q$. Here we summarize selected 
properties of the polynomials $p_{n}^{(q,\beta)}$. The generating function formula has already been formulated 
in~\eqref{eq:gener_func_alp=q}. As a consequence, one can deduce an explicit representation of $p_{n}^{(q,\beta)}$.

\begin{prop}
  \label{prop:explicit_repre}
  We have the explicit representation
  \[
    p_{n}^{(q,\beta)}(\upsilon(z))=\frac{(q;q)_{n}}{(\beta;q)_{n}}\sum_{k=0}^{n}\frac{(q\tau z^{-1};q)_{k}(q\tau^{-1}z;q)_{n-k}}{(q;q)_{k}(q;q)_{n-k}}z^{2k-n},
  \]
  and the $q$-hypergeometric representations
  \begin{equation}
    \label{eq:q-hyp_repre_alp=q}
    p_{n}^{(q,\beta)}(\upsilon(z))=\frac{(q\tau^{-1}z;q)_{n}}{z^{n}(\beta;q)_{n}}\pFq{2}{1}{q^{-n},q\tau z^{-1}}{q^{-n}\tau z^{-1}}{\tau z}
    =\frac{(q^{2};q)_{n}}{q^{n}\tau^{n}(\beta;q)_{n}}\pFq{3}{2}{q^{-n},q\tau z^{-1},q\tau z}{q^{2},0}{q},
  \end{equation}
  where $q\upsilon(\tau)=\beta\upsilon(z)$ and $\upsilon(z)=z+z^{-1}$.
\end{prop}

\begin{proof}
 By making use of the $q$-binomial formula~\eqref{eq:q-binom} on the right-hand side of~\eqref{eq:gener_func_alp=q},
 obtains
 \[
  \sum_{n=0}^{\infty}\frac{(\beta;q)_{n}}{(q;q)_{n}}p_{n}^{(\alpha,\beta)}(\upsilon(z);q)t^{n}=\pFq{1}{0}{q\tau z^{-1}}{-}{zt}\pFq{1}{0}{q\tau^{-1}z}{-}{z^{-1}t}
 \]
 for $|t|<|z|<1$. Multiplying the two $q$-hypergeometric series results in the expression
 \[
  \pFq{1}{0}{q\tau z^{-1}}{-}{zt}\pFq{1}{0}{q\tau^{-1}z}{-}{z^{-1}t}
  =\sum_{n=0}^{\infty}t^{n}\sum_{k=0}^{n}\frac{(q\tau z^{-1};q)_{k}(q\tau^{-1}z;q)_{n-k}}{(q;q)_{k}(q;q)_{n-k}}z^{2k-n}
 \]
 which yields the first formula from the statement. 
 
 To deduce the first $q$-hypergeometric representation of $p_{n}$, it suffices to note that
 \[
  \frac{(q\tau^{-1}z;q)_{n-k}}{(q;q)_{n-k}}=\frac{(q\tau^{-1}z;q)_{n}}{(q;q)_{n}}
  \frac{(q^{-n};q)_{k}}{(q^{-n}\tau z^{-1};q)_{k}}\tau^{k}z^{-k},
 \]
 for $k\in\{0,1,\dots,n\}$. The second $q$-hypergeometric representation follows from an application of the
 transformation formula~\eqref{eq:jack_term_qGaus_transf}.
\end{proof}

Except the obvious value of $p_{n}^{(\alpha,\beta)}$ at the origin, the first $q$-hypergeometric representation in~\eqref{eq:q-hyp_repre_alp=q} 
together with the $q$-Chu--Vandermonde identity~\eqref{eq:q-chu-vand} allows us to obtain certain special values 
of $p_{n}^{(q,\beta)}$.

\begin{cor}
 If $\beta(1+z^{2})=q^{2}+z^{2}$, then
 \[
  p_{n}^{(q,\beta)}(\upsilon(z);q)=\frac{(q^{2};q)_{n}}{z^{n}\left((q^{2}+z^{2})/(1+z^{2});q\right)_{n}}.
 \]
\end{cor}

The second $q$-hypergeometric representation can be used to relate $p_{n}^{(q,\beta)}$ with
a particular case of the continuous dual $q$-Hahn polynomials. Recall that the continuous dual $q$-Hahn polynomials
$p_{n}(x;a,b,c \mid q)$ are defined by~\cite[Eq.~(14.3.1)]{koekoek10}
\begin{equation}
 p_{n}(x;a,b,c \mid q)=a^{-n}(ab,ac;q)_{n}\,\pFq{3}{2}{q^{-n},ae^{\ii\theta},ae^{-\ii\theta}}{ab,ac}{q},
 \label{eq:def_cont_dual_qhahn}
\end{equation}
where $x=\cos\theta$. By comparing the formulas~\eqref{eq:q-hyp_repre_alp=q} and~\eqref{eq:def_cont_dual_qhahn},
one gets the relation
\[
 p_{n}^{(q,\beta)}(z;q)=\frac{1}{(\beta;q)_{n}}p_{n}\left(z/2;q\tau,q\tau^{-1},0\mid q\right)\!,
\]
where $q\upsilon(\tau)=\beta\upsilon(z)$ and $\upsilon(z)=z+z^{-1}$.

Next, let us investigate the measure of orthogonality of $p_{n}^{(q,\beta)}$. It turns out that the
absolutely continuous part of $\mu$ is expressible in terms of the $q$-Pochhammer symbols. On the other
hand, it seems that there is no significantly simpler expression for $\mu_{d}$ than the one from 
Proposition~\ref{prop:d_meas}. Nevertheless, the support of $\mu_{d}$, if non-empty, can be expressed fully
explicitly.

\begin{prop}
  \label{prop.d_meas_simplified}
 Let $\alpha=q$ and $\beta<1$. Then, for $\theta\in(0,\pi)$, one has
 \[
  \frac{\dd\mu_{ac}}{\dd\theta}(2\cos\theta) = -\frac{2(1-\beta)}{\pi(1-q)}\sin^{2}(\theta)\left|\frac{\left(q,qe^{2\ii\theta};q\right)_{\infty}}{\left(qe^{\ii\theta}\tau,qe^{\ii\theta}\tau^{-1};q\right)_{\infty}}\right|^{2},
 \]
 where $q(\tau+\tau^{-1})=2\beta\cos\theta$. Further, it holds
 \[
 \supp\mu_{d}=\left\{\pm\frac{q^{1-k}-q^{k+1}}{\sqrt{(\beta-q^{k+1})(q^{1-k}-\beta)}} \;\bigg|\; 
 k\in\N \mbox{ such that } \beta>\frac{q}{2}\left(q^{-k}+q^{k}\right) \right\}.
 \]
 Consequently, $\mu_{d}=0$ if and only if $\beta\leq(1+q^{2})/2$.
\end{prop}

\begin{proof}
 First, by using Proposition~\ref{prop:ac_meas} together with the $q$-Gauss summation formula~\eqref{eq:q-gauss_sum},
 one obtains
 \[
  \frac{\dd\mu_{ac}}{\dd\theta}(2\cos\theta)
  =-\frac{2(1-\beta)}{\pi(1-q)}\sin^{2}(\theta)\left|\frac{\left(q,qe^{2\ii\theta};q\right)_{\infty}}{\left(qe^{\ii\theta}\tau,qe^{\ii\theta}\tau^{-1};q\right)_{\infty}}\right|^{2}.
 \]
 
 Note that if $\beta\leq q<(1+q^{2})/2$, it readily follows from Proposition~\eqref{prop:mu_d_vanish_suf_cond} that $\mu_{d}=0$
 which is in agreement with the second part of the statement. Hence, we can further assume that $\beta\in(q,1)$.
 
 Next, we compute the zeros of $\psi^{+}_{-1}$ defined by~\eqref{eq:def_psi_pm}. These zeros can be found explicitly because
 \begin{align*}
  \psi_{-1}^{+}(z)&=z^{-1}\left(qz^{2};q\right)_{\!\infty}\pFq{2}{1}{z\tau,z\tau^{-1}}{qz^{2}}{q}
  =\frac{\left(qz\tau,qz\tau^{-1};q\right)_{\infty}}{z(q;q)_{\infty}}\\
  &=\frac{1}{z(q;q)_{\infty}}\prod_{k=1}^{\infty}\left(1-\beta(1+z^{2})q^{k-1}+z^{2}q^{2k}\right),
 \end{align*}
 by the $q$-Gauss sum~\eqref{eq:q-gauss_sum}. Consequently, for $\beta\in(q,1)$, the positive zeros of~$\psi^{+}_{-1}$ 
 are
 \[
  z_{k}=\sqrt{\frac{q^{1-k}-\beta}{\beta-q^{k+1}}}, \quad k\in\N.
 \]
 Note that $z_{1}\geq 1$ if and only if $\beta\leq(1+q^{2})/2$ and hence the latter condition implies that $\mu_{d}=0$
 according to Proposition~\ref{prop:d_meas}. If $\beta>(1+q^{2})/2$, then $z_{1}<1$ and the number of the zeros $z_{k}$
 located in $(0,1)$ is equal to the number of $k\in\N$ for which the inequality 
 \[
  \beta>\frac{q}{2}\left(q^{-k}+q^{k}\right)
 \]
 holds. Moreover, for these indices $k$, the support of $\mu_{d}$ consists of the points
 \[
  \pm\upsilon(z_{k})=\pm\frac{q^{1-k}-q^{k+1}}{\sqrt{(\beta-q^{k+1})(q^{1-k}-\beta)}}
 \]
 by Proposition~\ref{prop:d_meas}.
\end{proof}

%
%
\section{The kinetic Ising model with hyperbolically decaying temperature}
\label{sec:Ising}

Let us now return to the Ising model described briefly in Section~\ref{sec:Intro}.
For the convenience of the reader, we will begin this section with a summary of our results and then present their proofs in Subsections~\ref{subsec:magnetization} and~\ref{subsec:2spin}.

The current section focuses on an semi-infinite chain of particles indexed by $n \in \N$.
Each particle interacts only with its nearest neighbors and is in contact with a heath bath with hyperbolically decaying temperature $T_n = 2/(\kappa k_B n)$, where $\kappa > 0$, and so $\gamma_n = \tanh(\kappa n)$; see~\eqref{eq:gamma_n} and~\eqref{eq:gamma_tanh}.

Let again $q_n(t)$ denote the magnetization of the $n$th particle, $n\in\N$, see~\eqref{eq:magnetization}.
The time evolution of this quantity is governed by the following system of ordinary differential equations
\begin{equation}
  \label{eq:IsingODEs}
  \dot{q}_n(t) = -q_n(t) + \frac{\gamma_{n}}{2} \left(q_{n-1}(t) + q_{n+1}(t) \right), \quad n\in\N,
\end{equation}
accompanied with prescribed initial conditions $q_n(0) \in [-1,1]$, $n\in\N$, and a boundary condition $q_0(t) := 0$.
The factor $\gamma_n=\tanh(\kappa n)$ can be expressed in the form
\begin{equation}
  \label{eq:isingGammaN}
  \gamma_n = \frac{1 - e^{-2\kappa} e^{-2\kappa(n-1)}}{1 + e^{-2\kappa} e^{-2\kappa(n-1)}}, \quad n\in\N.
\end{equation}
Recalling Remark~\ref{rem:OriginalIsing}, i.e. comparing the expression~\eqref{eq:isingGammaN} with the one in~\eqref{eq:def_gamma_n_q_form}, we see that if we set
\[
  \alpha := q := e^{-2\kappa} \in (0, 1) \quad \text{and} \quad \beta := -q = -e^{-2\kappa} \in (-1, 0),
\]
then we can easily employ the results of Subsection~\ref{subsec:specialCase}.
For the purpose of this section let us set (see Proposition~\ref{prop:explicit_repre})
\begin{align}\label{eq.PnTheta}
  P_n(\theta) &:= \left( \prod_{i=1}^{n-1} \gamma_{i}\gamma_{i+1}\right)^{-1/2} p^{(q,-q)}_{n-1}\left(\upsilon(e^{\ii\theta});q\right) \\
  \nonumber
  &= \frac{\sqrt{\gamma_1}}{\sqrt{\gamma_n}} \sum_{k=0}^{n-1} \frac{(-q; q)_k (-q; q)_{n-1-k}}{(q; q)_{k}(q; q)_{n-1-k}} \left(e^{\ii\theta}\right)^{2k-n+1} \\
  \label{eq.PnSimplified}
  &= \frac{\sqrt{\gamma_1}}{\sqrt{\gamma_n}} \sum_{k=0}^{n-1} \frac{(-q; q)_k (-q; q)_{n-1-k}}{(q; q)_{k}(q; q)_{n-1-k}} \cos\left((2k-n+1)\theta\right), \quad n\in\N_0, \ \theta \in [0,\pi].
\end{align}

We are now able to formulate our main result concerning the magnetization.
For its proof, we refer the reader to Subsection~\ref{subsec:magnetization}.

\begin{thm}\label{thm:mag}
  \label{thm.magnetization}
  The unique solution $q_n(t)$ of~\eqref{eq:IsingODEs} satisfying the initial conditions $q_n(0) = q_{n,0}$, $n\in\N$, is given by
  \begin{equation}\label{eq:mag:superposition}
    q_n(t) = \sum_{k=1}^\infty q^{(k)}_n(t) q_{k,0},
  \end{equation}
  where $q_n^{(k)}(t)$ is the solution of~\eqref{eq:IsingODEs} satisfying $q_n^{(k)}(0) = \delta_{k,n}$.
  Furthermore, the solution $q_n^{(k)}(t)$ is given by the integral formula
  \begin{equation}
  \label{eq:magnetizationNK_final}
  q_n^{(k)}(t) = \frac{2}{\pi \gamma_1}\sqrt{\frac{\gamma_n}{\gamma_k}} \int_0^{\pi} e^{-t(1-\cos\theta)}\sin^2(\theta) P_n(\theta) P_k(\theta)\left|\frac{(q, q e^{2\ii\theta}; q)_\infty}{(-q, -q e^{2\ii\theta}; q)_\infty}\right|^2\dd\theta.
\end{equation}
  and obeys the asymptotic behavior
  \begin{equation}\label{eq:thm.magnetization.eq}
    q_n^{(k)}(t)=
    \frac{1}{\gamma_1}\sqrt{\frac{2}{\pi} \frac{\gamma_n}{\gamma_k}} \frac{(q; q)_\infty^4}{(-q; q)_\infty^4} \frac{P_{n}(0) P_{k}(0)}{t^{3/2}}+O\left(t^{-5/2}\right), \quad \text{as} \ t \to +\infty,
  \end{equation}
  where $P_n$ is defined in~\eqref{eq.PnTheta}.
\end{thm}

Note that the $q$-Pochhammer multiplicative factor appearing in~\eqref{eq:thm.magnetization.eq} is expressible in terms of the $\gamma_n$ factor in the following way
\[
  \frac{(q;q)_\infty}{(-q;q)_\infty} = \prod_{k=1}^\infty \gamma_k.
\]

Finally, we turn our attention to the evolution equation for the two-spin correlations $r_{m,n}(t)$, see~\eqref{eq:2SpinCorrelation}.
For $m,n\in\N$, $m \neq n$, the two-spin correlation $r_{m,n}(t)$ satisfies
\begin{equation}
  \label{eq:rmnOffDiagonal}
  \dot{r}_{m,n}(t) = -2r_{m,n}(t) + \frac{\gamma_m}{2} \left( r_{m+1,n}(t) + r_{m-1,n}(t) \right) + \frac{\gamma_n}{2} \left( r_{m,n+1}(t) + r_{m,n-1}(t) \right),
\end{equation}
where $r_{m,0} = r_{0,n} = 0$, $m,n\in\N$.
Moreover, it has to satisfy prescribed initial conditions $r_{m,n}(0)$ and
\begin{align}
  \label{eq:rmnDiagonal}
  r_{m,m}(t) &= 1, \quad m\in\N, \\
  \label{eq:rmnSymmetry}
  r_{m,n}(t) &= r_{n,m}(t), \quad m,n\in\N,
\end{align}
for all $t>0$.
In Subsection~\ref{subsec:2spin}, we will prove the following theorem describing some of the properties of the two-spin correlation. 

\begin{thm}\label{thm:two-spin_cor}
  There is a unique solution of~\eqref{eq:rmnOffDiagonal}--\eqref{eq:rmnSymmetry} given by
  \[
    r_{m,n}(t) = \rho_{m,n} + \sum_{k > \ell} r^{(k,\ell)}_{m,n}(t) (r_{k,\ell}(0) - \rho_{k,\ell}), \quad m,n\in\N,
  \]
  where $\rho_{m,n}$ is the stationary solution of~\eqref{eq:rmnOffDiagonal}--\eqref{eq:rmnSymmetry} and $r_{m,n}^{(k,\ell)}$, $k > \ell$, is the solution of~\eqref{eq:rmnOffDiagonal},~\eqref{eq:rmnSymmetry} and equation $r_{m,n}(t) = 0$.

  Furthermore, we have the following asymptotic expression
  \begin{equation}\label{eq:two-spin_corr_asympt}
    r_{m,n}^{(k,\ell)}(t) = \frac{3}{\pi\gamma_1^2} \frac{(q;q)_\infty^8}{(-q;q)_\infty^8} \sqrt{\frac{\gamma_m \gamma_n}{\gamma_k\gamma_\ell}} \frac{R_{k,\ell}R_{m,n}}{t^5} + O\left(t^{-6}\right), \quad \text{as} \ t \to +\infty, 
  \end{equation}
  for $m\geq n$, where, see~\eqref{eq.PnTheta},
  \begin{equation}\label{eq:def_R_mn}
    R_{m,n} = P_m(0) P''_n(0) - P''_m(0) P_n(0).
  \end{equation}
\end{thm}

In conclusion, we see that the two-spin correlation approaches the stationary solution at the rate proportional to $t^{-5}$.
This result is in contrast with the classical case of the doubly-infinite chain with constant temperature~\cite{glauber_jmp63}, where the rate is exponential, cf. Remark~\ref{rem:magnetizationConstantT}.
Moreover, in the classical case, the stationary solution $\rho_{m,n}$ is known fully explicitly.
Whether the stationary solution in the current case under investigation is at least expressible in terms of the basic hypergeometric series remains an open problem.

%
%
\subsection{The magnetization}
\label{subsec:magnetization}

The family $\{P_n\}_{n=1}^{\infty}$ defined in~\eqref{eq.PnTheta} forms an orthonormal basis of the Hilbert space $L^2\left((0,\pi), \dd\mu\right)$, where the measure $\mu$ is absolutely continuous with respect to the Lebesgue measure and (see Proposition~\ref{prop.d_meas_simplified})
\begin{equation}\label{eq.muSimplified}
  \frac{\dd\mu}{\dd\theta}(\theta) = \frac{2}{\pi\gamma_1} \sin^2(\theta)
  \left| 
  \frac{(q, q e^{2\ii\theta}; q)_\infty}{(-q, -q e^{2\ii\theta}; q)_\infty}
  \right|^2.
\end{equation}
Moreover, $\{P_n(\theta)\}_{n=1}^{\infty}$ satisfies the following symmetric recurrence equation
\[
  \frac{\sqrt{\gamma_{n-1}\gamma_n}}{2} P_{n-1}(\theta) + \frac{\sqrt{\gamma_n\gamma_{n+1}}}{2} P_{n+1}(\theta) = \cos(\theta) P_{n}(\theta), \quad n\in\N.
\]
Note that $P_{0}(\theta) = 0$.
Also observe that $P_n(0)$ is positive and $P'_n(0)$ vanishes for all $n\in\N$.

Now consider the solution $q_n^{(k)}(t)$ of the magnetization evolution equation~\eqref{eq:IsingODEs} with initial conditions $q_n^{(k)}(0) = \delta_{k,n}$, $n\in\N$, for some fixed $k\in\N$.
Such a solution is unique and can be represented in the following integral form
\begin{equation}
  \label{eq:magnetizationNK}
  q_n^{(k)}(t) = \sqrt{\frac{\gamma_n}{\gamma_k}} \int_0^{\pi} e^{-t(1-\cos\theta)} P_n(\theta) P_k(\theta) \,\dd\mu(\theta).
\end{equation}
The multiplicative factor comes from the symmetrization matrix $\Gamma=\diag(\sqrt{\gamma_{1}},\sqrt{\gamma_{2}},\dots)$, see~\eqref{eq:magnetizationODEsolution}, where $N$ is sent to infinity. Formulas \eqref{eq.muSimplified} and \eqref{eq:magnetizationNK} yield \eqref{eq:magnetizationNK_final}.

A simple graphical illustration of solutions~\eqref{eq:magnetizationNK} is presented in Figure~\ref{fig:qs}.
The general solution to the equation~\eqref{eq:IsingODEs} can be obtained by suitable superposition~\eqref{eq:mag:superposition} of solutions~\eqref{eq:magnetizationNK}.
The convergence of the series~\eqref{eq:mag:superposition} is guaranteed by the following Lemma. Recall that the sequence $\{q_{k,0}\}_{k=1}^{\infty}$ appearing in~\eqref{eq:mag:superposition} is bounded since $q_{k,0}\in[-1,1]$ for all $k\in\N$.

\begin{figure}
  {\centering
    \includegraphics{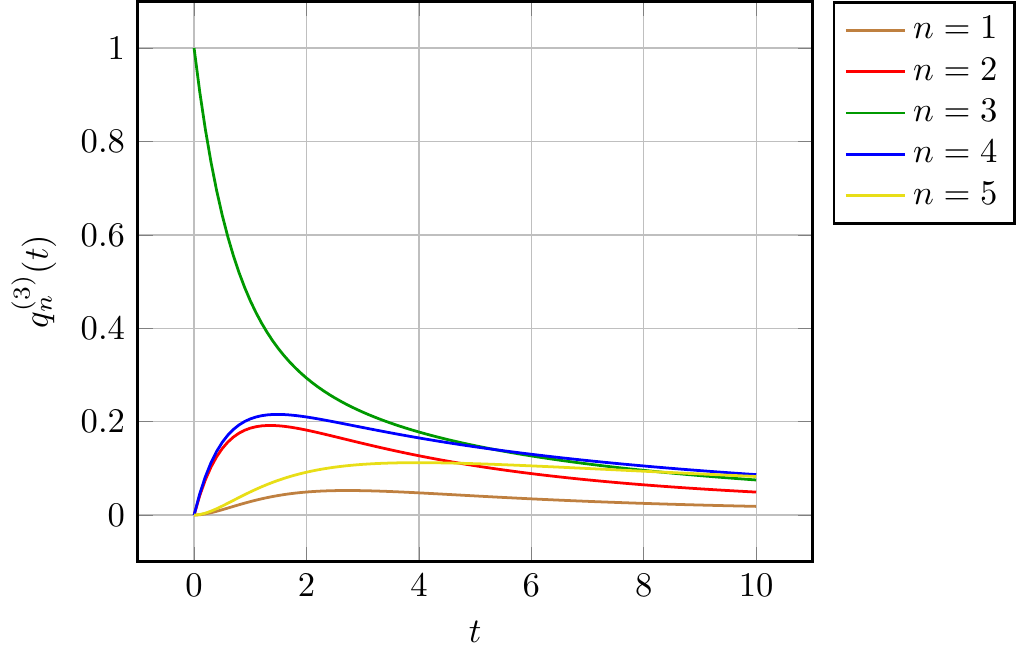}
  \par}
  \caption{Magnetization $q_n^{(3)}(t)$ given by~\eqref{eq:magnetizationNK} for $n=1,2,\ldots,5$.}
  \label{fig:qs}
\end{figure}

\begin{lem}\label{lem:mag_bound}
  Let $\ell\in\N$. Then the magnetization $q_n^{(k)}(t)$, given by~\eqref{eq:magnetizationNK}, satisfies the estimate
  \[
    \left| q_n^{(k)}(t) \right| \leq \frac{C_{n,\ell}(q,
    t)}{k^{\ell}}, \quad \forall k\in\N.
  \]
  The factor $C_{n,\ell}(q, t)$ does not depend on $k$.
\end{lem}

\begin{proof}
  In order to avoid unnecessary and cumbersome expressions, we will only provide a sketch of the proof.

  First of all, note that both $P_n(\theta)$ and $\mathrm{d}\mu(\theta) / \mathrm{d}\theta$ can be extended to even, smooth, and $2\pi$-periodic functions w.r.t.~$\theta$, see equations~\eqref{eq.PnSimplified} and~\eqref{eq.muSimplified}, respectively.
  From the very definitions of the $q$-hypergeometric series and the $q$-Pochhammer symbol we obtain the asymptotic expansion
  \begin{equation}\label{eq:qkn_proof_2}
    \pFq{2}{1}{-1, -e^{-2\ii\theta}}{qe^{-2\ii\theta}}{q^{k+1}} = 1 + O(q^k),
  \end{equation}
  which is uniform w.r.t~$\theta$ since the error term is majorized by a factor independent of~$\theta$.
  Combining these observations with \eqref{eq.PnTheta}, \eqref{eq.muSimplified}, and \eqref{eq:ogp_formula_psipm}, we are able to express $q_n^{(k)}(t)$ from equation~\eqref{eq:magnetizationNK} in the following form
  \begin{equation}\label{eq:qkn_proof_1}
    q_n^{(k)}(t) = \tilde{C}_{k,n}(q) \left( \Im F_{n,k}(q,t) + G_{n,k}(q,t)\right),
  \end{equation}
  where
  \[
 F_{n,k}(q,t):=\frac{1}{2\pi}\int_{-\pi}^\pi e^{-\ii k\theta} f_n(q,t;\theta) \,\mathrm{d}\theta
  \]
  and $f_n(q,t;\theta)$ is a smooth and $2\pi$-periodic function of $\theta$.
  The term $F_{n,k}(q,t)$ can be thought of as being the $k$th Fourier coefficient of $f_n(t,q;\theta)$.
  It is well known (e.g., cf.~\cite[Chapter II, Miscellaneous theorems and examples, no. 5]{Zygmund}) that Fourier coefficients of a smooth $2\pi$-periodic function decay as $1/k^\ell$ for any fixed positive integer $\ell$.
  
  The term $G_{n,k}(q,t)$ in~\eqref{eq:qkn_proof_1} originating from the second term of~\eqref{eq:qkn_proof_2}, can be bounded as
  \[
    |G_{n,k}(q,t)|\leq\tilde{\tilde{C}}_n(q,t) q^{k}.
  \]
  Finally, the multiplicative factor $\tilde{C}_{k,n}(q)$ in~\eqref{eq:qkn_proof_1} is bounded by a constant independent of $k$ and $n$.
\end{proof}

It is unlikely that~\eqref{eq:magnetizationNK} can be integrated explicitly.
In the next proposition, we therefore derive the asymptotic expression from Theorem~\ref{thm.magnetization}.
Moreover, we derive also higher order terms which are quite complicated, however, they are necessary in order to determine the leading term  of the asymptotic expansion for the two-spin correlation which is computed in Subsection~\ref{subsec:2spin} below.

\begin{prop}\label{prop.qnkAsy}
  Let $n,k\in\N$.
  Then $q_n^{(k)}(t)$ given by~\eqref{eq:magnetizationNK} has the following asymptotic expansion
  \begin{equation}\label{eq.qnkAsy}
    q_n^{(k)}(t) = \frac{1}{\gamma_1} \sqrt{\frac{2}{\pi} \frac{\gamma_n}{\gamma_k}}\frac{(q;q)_\infty^4}{(-q;q)_\infty^4}P_{n}(0)P_{k}(0) \, t^{-3/2} 
    \left( 1 + \frac{3}{2} A_{n,k} \, t^{-1} + \frac{15}{4} B_{n,k} \, t^{-2} + O\left(t^{-3}\right) \right),
  \end{equation}
  where
  \begin{align}
    \label{eq:Ank}
    A_{n,k} &= \frac{P''_n(0)}{P_n(0)} + \frac{P''_k(0)}{P_k(0)} + 8\phi_1 + 16\phi_2-\frac{1}{4}, \\\nonumber
    B_{n,k} &= \frac{P''_n(0) P''_k(0)}{P_n(0)P_k(0)} + \left(8\phi_1 + 16\phi_2 - \frac{1}{12}\right)\left(\frac{P''_n(0)}{P_n(0)} + \frac{P''_k(0)}{P_k(0)} \right)
    +\frac{1}{6} \left(\frac{P^{(4)}_n(0)}{P_n(0)} + \frac{P^{(4)}_k(0)}{P_k(0)} \right)  \\\nonumber
    &- \frac{16}{3} \phi_1^4 + 64\phi_1^2\phi_2 + 64\phi_2^2 + 128\phi_1\phi_2  -128\phi_1\phi_3 + 32\phi_1^2 - 6\phi_1 - 76\phi_2 -192\phi_3-\frac{1}{32}, \\\label{eq:Bnk}    
  \end{align}
  as $t \to +\infty$, where the constants $\phi_1$, $\phi_2$, and $\phi_3$ are given by
  \begin{equation}\label{eq:def_phi123}
    \phi_1 := \sum_{j=1}^\infty \frac{2 q^{j}}{1 - q^{2j}}, \quad
    \phi_2 := \sum_{j=1}^\infty \frac{2 q^{3j}}{\left( 1 - q^{2j} \right)^2},
    \quad \text{and} \quad
    \phi_3 := \frac{2}{3} \sum_{j=1}^\infty \frac{q^{3j}(1 + 3 q^{2j})}{\left( 1 - q^{2j} \right)^3}.
  \end{equation}

\end{prop}

\begin{proof}
  The integral in~\eqref{eq:magnetizationNK} can be written in the following form
  \[
    q_n^{(k)}(t) = \frac{2}{\pi\gamma_1} \sqrt{\frac{\gamma_n}{\gamma_k}} \int_0^2 e^{-tx} x^{1/2} h(x) \,\mathrm{d}x,
  \]
  where
  \[
    h(x) = \sqrt{2-x} \, P_n\big(\theta(x)\big) P_k\big(\theta(x)\big) \left| \frac{\big(q,q e^{2\ii\theta(x)};q\big)_\infty}{\big(-q,-q e^{2\ii\theta(x)};q\big)_\infty} \right|^2
  \]
  and $\theta(x) = \arccos(1-x)$.
  In order to apply Watson's Lemma, see~\cite{olver_97}, we have to compute the asymptotic expansion of $h(x)$ as $x \to 0$.
  Expanding factors of $h(x)$ we obtain
  \begin{align*}
    \sqrt{2 - x} &= \sqrt{2}\left( 1 - \frac{x}{4} - \frac{x^2}{32} +  O(x^3)  \right), \\
    P_n(\theta(x)) &= P_n(0) + P''_n(0) x + \frac{1}{6} \left( P''_n(0) + P^{(4)}_n(0) \right)  x^2 + O(x^3),
  \end{align*}
  as $x\to0$.
  Let us turn our attention to the last factor.
  A lengthy but straightforward computation yields the expansion
  \begin{align*}
    \frac{(qz;q)_\infty}{(-qz;q)_{\infty}} = \frac{(q;q)_\infty}{(-q;q)_\infty} \bigg[1 &- \phi_1 (z-1) + \left(\frac{1}{2}\phi_1^2 -\phi_2\right)(z-1)^2 \\
    &\hskip24pt +\left( \phi_1\phi_2 - \phi_3 - \frac{1}{6} \phi_1^3 \right) (z-1)^3 \bigg] + O\left((z-1)^4\right), \ \text{as} \ z \to 1,
  \end{align*}
  where the constants $\phi_{1}$, $\phi_{2}$, and $\phi_{3}$ are defined by~\eqref{eq:def_phi123}. Consequently, we have
  \begin{align*}
    &\left| \frac{(q,qe^{2\ii\theta(x)};q)_{\infty}}{(-q,-qe^{2\ii\theta(x)};q)_{\infty}} \right|^2 = \frac{(q;q)_\infty^4}{(-q;q)_{\infty}^{4}} \bigg(1 + 8(\phi_1 + 2\phi_2) x \\
    &- \frac{4}{3} \left( 3\phi_1 + 54\phi_2 + 144 \phi_3 - 24\phi_1^2 -48\phi_2^2 - 96\phi_1 (\phi_2 - \phi_3) - 48\phi_1^2\phi_2 + 4\phi_1^4 \right) x^2 \bigg) + O\left(x^3\right),
  \end{align*}
  as $x \to 0$.
  Putting all these expansions together we conclude that
  \[
    h(x) = \sqrt{2} \frac{(q;q)_\infty^4}{(-q;q)_\infty^4} \bigg[ P_n(0) P_k(0) + A_{n,k} \, x + B_{n,k} \, x^2 \bigg] + O(x^3), \quad \text{as} \ x \to 0,
  \]
  where $A_{n,k}$ and $B_{n,k}$ are given by~\eqref{eq:Ank} and~\eqref{eq:Bnk}, respectively.
 This computation was verified by the computer algebra system Wolfram Mathematica.

Finally, using Watson's Lemma we immediately obtain the asymptotic behavior~\eqref{eq.qnkAsy} of $q_n^{(k)}(t)$ as $t\to\infty$.
\end{proof}

It is interesting to compare this decay rate with the one exhibited by the model with constant temperature heath baths.

\begin{rem}
  \label{rem:magnetizationConstantT}
  In case of constant temperature $T > 0$ we have
  \[
    \gamma = \tanh\left(\frac{2}{k_B T}\right) \in (0,1),
  \]
  and the magnetization obeys the following differential equation
  \begin{equation}
    \label{eq:diffT}
    \dot{q}_n(t) = -q_n(t) + \frac{\gamma}{2}\left(q_{n-1}(t) + q_{n+1}(t)\right), \quad n\in\N,
  \end{equation}
  where $q_{0}(t) = 0$.
  Following a similar procedure to the one described in the previous sections, it is straightforward to see that
  \[
    p_n(\theta) = \frac{\sin n\theta}{\sin\theta}, \quad n\in\N, \ \theta \in (0,\pi),
  \]
  satisfies $p_{0}(\theta) = 0$, $p_{1}(\theta) = 1$, and the recurrence equation
  \[
    -p_n(\theta) + \frac{\gamma}{2}\big( p_{n-1}(\theta) + p_{n+1}(\theta)\big) = (-1 + \gamma\cos\theta) p_{n}(\theta), \quad n\in\N.
  \]
  The family $\{p_{n}(\theta)\}_{n=1}^\infty$ forms an orthonormal basis of the Hilbert space $L^2\big((0,\pi),\frac{2}{\pi}\sin^2(\theta)\mathrm{d}\theta\big)$.

  Consequently, the solution of the differential equation~\eqref{eq:diffT} satisfying $q_{n}^{(k)}(0) = \delta_{n,k}$ is given by the following expression
  \begin{align*}
    q_n^{(k)}(t) &= \frac{2}{\pi} \int_0^\pi e^{(-1+\gamma\cos\theta)t} \sin(n\theta) \sin(k\theta)\,\mathrm{d}\theta = \\
    &= e^{-t} \left( I_{n-k}(\gamma t) - I_{n+k}(\gamma t) \right),
  \end{align*}
  where $I_n(z)$ is the modified Bessel function of the first kind, see~\cite[Eq.~(10.32.3)]{dlmf}.
  Using the well known asymptotic behavior of $I_n(z)$ for large argument~\cite[Eq.~(10.40.1)]{dlmf}, we arrive at the asymptotic expansion
  \begin{equation}
    \label{eq:qnkAsyConstant}
    q_n^{(k)}(t)=\sqrt{\frac{2}{\pi}}e^{-t(1-\gamma)}\left(\frac{nk}{(\gamma t)^{3/2}}+O\left(t^{-5/2}\right)\right), \quad \text{as} \ t\to+\infty.
  \end{equation}
  A similar behavior with exponential decay is exhibited in the case of the doubly-infinite spin chain considered by Glauber in~\cite{glauber_jmp63}. It is worth noticing that~\eqref{eq.qnkAsy}, in contrast to~\eqref{eq:qnkAsyConstant}, does not demonstrate the exponential decay.
\end{rem}

%
%
\subsection{Two-spin correlations}
\label{subsec:2spin}

In order to solve the (non-homogeneous) problem~\eqref{eq:rmnOffDiagonal}--\eqref{eq:rmnSymmetry} we proceed in three steps.

First of all, let $\{\rho_{m,n}\}_{m,n=1}^\infty$ be the corresponding stationary solution, i.e. a solution of equations
\begin{align}
  \label{eq:stationary1}
  -2\rho_{m,n} + \frac{\gamma_m}{2} \left( \rho_{m+1,n} + \rho_{m-1,n} \right) + \frac{\gamma_n}{2} \left( \rho_{m,n+1} + \rho_{m,n-1} \right)&=0, \quad m,n\in\N, \ m \neq n, \\
  \label{eq:stationary2}
  \rho_{m,m} &= 1, \quad m\in\N,
\end{align}
and $\rho_{m,0} = \rho_{0,n} = 0$ for $m,n\in\N$.
Further, we prove the stationary solution exists and is unique. To see this, consider a simple reformulation of the problem~\eqref{eq:stationary1}--\eqref{eq:stationary2},
\[
  \rho = T\rho,
\]
where $T$ is a non-linear mapping acting on $\ell^\infty \equiv \ell^\infty(\N^{2})$ -- the Banach space  of bounded sequences indexed by elements of $\N^{2}$ with the norm $\|x\|_\infty= \sup_{m,n}|x_{m,n}|$ --  defined by
\[
  (Tx)_{m,n} = \begin{cases}
    1, & m=n, \\
    \frac{\gamma_m}{4} \left(x_{m+1,n} + x_{m-1},n\right) + \frac{\gamma_n}{4} \left(x_{m,n+1} + x_{m,n-1}\right), & m\neq n,
  \end{cases}
\]
where $x_{0,n} = x_{m,0} = 0$, $m,n\in\mathbb{N}$, for each $x \in \ell^\infty$.
The operator $T$ is a contraction, in particular the inequality
\[
  \|Tx - Ty\|_\infty \leq \delta \|x - y\|_\infty
\]
holds true for any $x,y \in \ell^\infty$ with 
\[
\delta = \max_{n\in\N}|\gamma_{n}| = \frac{1-q}{1+q} \in (0,1).
\]
The well-known Banach fixed point theorem now guarantees the existence of a unique $\rho \in \ell^\infty$ satisfying $\rho = T\rho$.  

Next, we construct a special symmetric solution $r_{m,n}$ of~\eqref{eq:rmnOffDiagonal} with vanishing diagonal, i.e. satisfying $r_{m,m} = 0$.
For this purpose we will use the magnetization, specifically the solution~\eqref{eq:magnetizationNK}.
For $k > \ell$, set
\begin{equation}\label{eq:rmn}
  r_{m,n}^{(k,\ell)}(t) := \begin{cases}
    q_m^{(k)}(t) q_n^{(\ell)}(t) - q_n^{(k)}(t) q_m^{(\ell)}(t), & m \geq n, \\
    q_m^{(\ell)}(t) q_n^{(k)}(t) - q_n^{(\ell)}(t) q_m^{(k)}(t), & m < n.
  \end{cases}
\end{equation}
Employing the results of the previous subsection, i.e. properties of $q_n^{(k)}$, we see that $r_{m,n}^{(k,\ell)}$ enjoys the following properties
\begin{align*}
  \frac{\mathrm{d} r_{m,n}^{(k,\ell)}}{\mathrm{d}t}(t) &=
    -2 r_{m,n}^{(k,\ell)}(t) + \frac{\gamma_m}{2} \left( r_{m+1,n}^{(k,\ell)}(t) + r_{m-1,n}^{(k,\ell)}(t) \right) + \frac{\gamma_n}{2} \left( r_{m,n+1}^{(k,\ell)}(t) + r_{m,n-1}^{(k,\ell)}(t) \right), \\
  r_{m,n}^{(k,\ell)}(t) &= r_{n,m}^{(k,\ell)}(t), \\
  r_{m,n}^{(k,\ell)}(0) &= \delta_{m,k}\delta_{n,\ell}, \ \text{for} \ m\neq n,\\
  r_{m,m}^{(k,\ell)}(t) &= r_{0,n}^{(k,\ell)}(t) = r_{m,0}^{(k,\ell)}(t) = 0,
\end{align*}
where $k,\ell,m,n\in\N$, $k>\ell$, and $t>0$.

Next, we verify the asymptotic formula~\eqref{eq:two-spin_corr_asympt}. To this end, we use the definition~\eqref{eq:rmn} for $m\geq n$ and the asymptotic expansion of magnetization 
$q_n^{k}(t)$, for $t\to+\infty$, given by Proposition~\ref{prop.qnkAsy}. It is not difficult to check that the coefficients corresponding to $t^{-3}$ and $t^{-4}$ in the expansion of $r_{m,n}^{(k,\ell)}(t)$ are both vanishing. The next coefficient turns out to be the leading term of the expansion, which reads
\begin{align}
r_{m,n}^{(k,\ell)}(t)&=\frac{2}{\pi\gamma_1^2}\sqrt{\frac{\gamma_m \gamma_n}{\gamma_k\gamma_\ell}} \frac{(q;q)_\infty^8}{(-q;q)_\infty^8}P_{k}(0)P_{\ell}(0)P_{m}(0)P_{n}(0)\,t^{-5}\nonumber\\
&\times\left(\frac{9}{4}\left(A_{m,k}A_{n,\ell}-A_{n,k}A_{m,\ell}\right)+\frac{15}{4}\left(B_{m,k}+B_{n,\ell}-B_{n,k}-B_{m,\ell}\right) + O\left(t^{-1}\right)\right),\label{eq:tow_asympt_rmn_inproof}
\end{align}
for $t\to+\infty$.
By inspection of the terms from~\eqref{eq:Ank} and~\eqref{eq:Bnk}, one finds out that most of them do not contribute to~\eqref{eq:tow_asympt_rmn_inproof}. In fact, the non-trivial expression for $A_{m,k}A_{n,\ell}-A_{n,k}A_{m,\ell}$ equals
\begin{align*}
&\left(\frac{P''_m(0)}{P_m(0)} + \frac{P''_k(0)}{P_k(0)}\right)\left(\frac{P''_n(0)}{P_n(0)} + \frac{P''_\ell(0)}{P_\ell(0)}\right)-
\left(\frac{P''_m(0)}{P_m(0)} + \frac{P''_k(0)}{P_k(0)}\right)\left(\frac{P''_n(0)}{P_n(0)} + \frac{P''_\ell(0)}{P_\ell(0)}\right)\\
&=-\frac{R_{k,\ell}R_{m,n}}{P_{k}(0)P_{\ell}(0)P_{m}(0)P_{n}(0)},
\end{align*}
and the expression $B_{m,k}+B_{n,\ell}-B_{n,k}-B_{m,\ell}$ differs only in sign from the previous one since it equals
\[
\frac{P''_m(0)}{P_m(0)}\frac{P''_k(0)}{P_k(0)}+\frac{P''_n(0)}{P_n(0)}\frac{P''_\ell(0)}{P_\ell(0)}-\frac{P''_n(0)}{P_n(0)}\frac{P''_k(0)}{P_k(0)}-\frac{P''_m(0)}{P_m(0)}\frac{P''_\ell(0)}{P_\ell(0)}
=\frac{R_{k,\ell}R_{m,n}}{P_{k}(0)P_{\ell}(0)P_{m}(0)P_{n}(0)},
\]
where the coefficients $R_{m,n}$ are given by~\eqref{eq:def_R_mn}. Plugging the last two expressions into~\eqref{eq:tow_asympt_rmn_inproof}, one arrives at the asymptotic formula~\eqref{eq:two-spin_corr_asympt}.

Finally, the general solution of~\eqref{eq:rmnOffDiagonal}--\eqref{eq:rmnSymmetry} with initial data $r_{m,n}(0)$, $m,n\in\N$, is given by the superposition
\[
  r_{m,n}(t) = \rho_{m,n} + \sum_{k,\ell\in\N,\, k>\ell} \left( r_{k,\ell}(0) - \rho_{k,\ell} \right)r_{m,n}^{(k,\ell)}(t),
\]
whose convergence is guaranteed by Lemma~\ref{lem:mag_bound} and boundedness of $\left\{r_{k,\ell}(0) - \rho_{k,\ell}\right\}_{k,\ell\in\N}$.

\subsection*{Acknowledgment}

T.~K. acknowledges the financial support by the Ministry of Education, Youth and Sports of the Czech Republic project no.~CZ.02.1.01/0.0/0.0/16\_019/0000778. The research of F.~{\v S}. was supported by the GA{\v C}R grant No.~20-17749X.

%
%

\appendix
\section{Basic hypergeometric series}\label{app:A}

Since the theory of basic hypergeometric functions is perhaps less well-known in contrast to its classical 
counterpart, we briefly summarize basic definitions and selected identities that are needed in the text. 
All details and much more can be found in~\cite{gasper90}.

The parameter $q$ is always assumed to satisfy $0<q<1$. For $a\in\C$ and $n\in\N_{0}\cup\{\infty\}$, 
the $q$-shifted factorials are defined by
\[
 (a;q)_{n}:=\prod_{j=0}^{n-1}\left(1-aq^{j}\right).
\]
They can be also defined for negative values of $n$ as
\[
 (a;q)_{n}:=\prod_{j=n}^{-1}\left(1-aq^{j}\right)^{-1},
\]
provided that $a\notin\{q,q^{2},\dots,q^{-n}\}$.

For $r,s\in\N_{0}$ and $a_{1},\dots,a_{r},b_{1},\dots,b_{s}\in\C$, the basic hypergeometric series is defined by
\[
 {}_{r}\phi_{s}\left(\begin{matrix}a_{1},\dots,a_{r}\\
            b_{1},\dots,b_{s} \end{matrix}\bigg|q;z\right)
 :=\sum_{k=0}^{\infty}\frac{(a_{1},\dots,a_{r};q)_{k}}{(b_{1},\dots,b_{s};q)_{k}}(-1)^{(s-r+1)k}q^{(s-r+1)k(k-1)/2}\frac{z^{k}}{(q;q)_{k}},
\]
where
\[
 (a_{1},\dots,a_{s};q)_{k}=(a_{1};q)_{k}\dots(a_{s};q)_{k},
\]
provided that $b_{j}\notin q^{-\N_{0}}$ for all $1\leq j\leq s$. If one of the numerator parameters $a_{j}$ equals
$q^{-n}$, for some $n\in\N_{0}$, the basic hypergeometric series is a polynomial in $z$. Otherwise the radius of convergence
$\rho=\rho(r,s)$ of the basic hypergeometric series is equal to
\[
 \rho=\begin{cases}
       \infty, & \mbox{ if } r<s+1,\\
       1, & \mbox{ if } r=s+1,\\
       0, & \mbox{ if } r>s+1.
      \end{cases}
\]

The basic hypergeometric series ${}_{r}\phi_{s}$ represents a $q$-analogue, i.e., a certain one-parameter generalization
of the classical hypergeometric series ${}_{r}F_{s}$ since
\[
 \lim_{q\to1-}{}_{r}\phi_{s}\left(\begin{matrix}q^{a_{1}},\dots,q^{a_{r}}\\q^{b_{1}},\dots,q^{b_{s}} \end{matrix}\bigg|q;(q-1)^{s-r+1}z\right)=
 {}_{r}F_{s}\left(\begin{matrix}a_{1},\dots,a_{r}\\b_{1},\dots,b_{s} \end{matrix}\bigg|z\right).
\]
The function ${}_{2}\phi_{1}$, which is a $q$-analogue of the Gauss hypergeometric function ${}_{2}F_{1}$, is called the 
\emph{$q$-Gauss hypergeometric series}.

Next, we list several selected identities used above. The parameters are always assumed to be such that
all the involved basic hyperbolic series are well defined. The $q$-binomial theorem~\cite[Eq.~(II.~3)]{gasper90}:
\begin{equation}
 \pFq{1}{0}{a}{-}{z}=\frac{(az;q)_{\infty}}{(z;q)_{\infty}}.
 \label{eq:q-binom}
\end{equation}
The $q$-Chu--Vandermonde summation identity~\cite[Eq.~(II.~6)]{gasper90}:
\begin{equation}
 \pFq{2}{1}{q^{-n},a}{c}{q}=\frac{(a^{-1}c;q)_{n}}{(c;q)_{n}}a^{n}.
 \label{eq:q-chu-vand}
\end{equation}
The $q$-Gauss summation formula~\cite[Eq.~(II.~6)]{gasper90}:
\begin{equation}
\pFq{2}{1}{a,b}{c}{\frac{c}{ab}}=\frac{\left(a^{-1}c,b^{-1}c;q\right)_{\infty}}{\left(c,a^{-1}b^{-1}c;q\right)_{\infty}}.
 \label{eq:q-gauss_sum}
\end{equation}
Jackson's transformation of terminating $q$-Gauss hypergeometric~\cite[Eq.~(III.~7)]{gasper90}:
\begin{equation}
 \pFq{2}{1}{q^{-n},b}{c}{z}=\frac{(cb^{-1};q)_{n}}{(c;q)_{n}}\pFq{3}{2}{q^{-n},b,q^{-n}bc^{-1}z}{q^{1-n}bc^{-1},0}{q}.
 \label{eq:jack_term_qGaus_transf}
\end{equation}

\section{Properties of zeros of certain Gauss \texorpdfstring{$q$}{q}-hypergeometric series}\label{app:B}

We deduce certain properties of zeros of functions $\psi_{n}^{+}$ defined by~\eqref{eq:def_psi_pm}
that were needed in the proof of Proposition~\ref{prop:d_meas}.

\begin{lem}\label{lem:aux_green_id}
 For $|\alpha|<1$, $\beta<1$, $n\in\N_{0}$, and $0<|z|<1$, one has
 \[
  \sum_{k=n}^{\infty}\frac{1-\beta q^{k}}{1-\alpha q^{k}}\left(\psi_{k}^{+}(z)\right)^{2}=\frac{z^{2}}{z^{2}-1}\left[\left(\psi_{n-1}^{+}\right)'(z)\psi_{n}^{+}(z)-\psi_{n-1}^{+}(z)\left(\psi_{n}^{+}\right)'(z)\right].
 \]
\end{lem}

\begin{proof}
 According to Proposition~\ref{prop:two_sol_psi_pm}, $\psi^{+}$ solves the equation~\eqref{eq:diff_eq} which means that
 \begin{equation}
  \psi_{n-1}^{+}(z)-\upsilon(z)\gamma_{n}^{-1}\psi_{n}^{+}(z)+\psi_{n+1}^{+}(z)=0,
 \label{eq:diff_eq_psi_plus}
 \end{equation}
 for all $n\in\N_{0}$ and $0<|z|<1$, where $\gamma_{n}$ is as in~\eqref{eq:def_gamma_n_q_form}.
 From~\eqref{eq:diff_eq_psi_plus}, one deduces
 \[
  (\upsilon(z)-\upsilon(\tilde{z}))\gamma_{n}^{-1}\psi_{n}^{+}(z)\psi_{n}^{+}(\tilde{z})=W_{n-1}\left(\psi^{+}(z),\psi^{+}(\tilde{z})\right)-W_{n}\left(\psi^{+}(z),\psi^{+}(\tilde{z})\right),
 \]
 for all $0<|z|<1$, $0<|\tilde{z}|<1$, and $n\in\N_{0}$. With the later restrictions, $W_{n}\left(\psi^{+}(z),\psi^{+}(\tilde{z})\right)\to0$, for $n\to\infty$, as it follows from~\eqref{eq:asympt_psipm_n_infpos}. 
 Hence, by summing up the above equations, one gets
 \[
  (\upsilon(z)-\upsilon(\tilde{z}))\sum_{k=n}^{\infty}\gamma_{k}^{-1}\psi_{k}^{+}(z)\psi_{k}^{+}(\tilde{z})=W_{n-1}\left(\psi^{+}(z),\psi^{+}(\tilde{z})\right), \quad n\in\N_{0}.
 \]
 Finally, by dividing both sides of the above equation by $\upsilon(z)-\upsilon(\tilde{z})$ and sending $\tilde{z}\to z$, one arrives at the identity from the statement.
\end{proof}

\begin{prop}\label{prop:zeros_basic_propert}
 Let $\alpha\in(-1,1)$, $\beta<1$, and $n\geq-1$. Then the function $\psi_{n}^{+}$ is analytic in the punctured unit disc $\mathbb{D}\setminus\{0\}$ and its zeros in $\mathbb{D}\setminus\{0\}$ are all real, simple, 
 and symmetrically distributed with respect to the origin. Moreover, between any two neighboring zeros of $\psi_{n}^{+}$ in $(0,1)$ there is exactly one zero of $\psi_{n+1}^{+}$.
\end{prop}

\begin{proof}
 Let $n\geq-1$ be fixed. The fact that $\psi_{n}^{+}$ is analytic in $\mathbb{D}\setminus\{0\}$ 
 has already been noted in Remark~\ref{rem:basic_prop_psi_pm} and follows immediately from the definition~\eqref{eq:def_psi_pm}.
 
 First, we show the that all possible zeros of $\psi_{n}^{+}$ in $\mathbb{D}\setminus\{0\}$ have to be real. Let $z_{0}\in\mathbb{D}\setminus\{0\}$ be a~zero of $\psi_{n}^{+}$. Put 
 \[
  u_{k}:=\gamma_{n+k}^{-1/2}\psi_{n+k}^{+}(z_{0}), \quad k\in\N_{0},
 \]
 where $\gamma_{k}$ is as in~\eqref{eq:def_gamma_n_q_form}. By using~\eqref{eq:diff_eq}, one verifies that $u$ satisfies equations
 \[
 \sqrt{\gamma_{n+k-1}\gamma_{n+k}}u_{k-1}-\upsilon(z_{0})u_{k}+\sqrt{\gamma_{n+k}\gamma_{n+k+1}}u_{k+1}=0, \quad k\in\N.
 \]
 Note also that the vector $u=(u_{1},u_{2},\dots)$ belong to $\ell^{2}(\N)$ as it follows from~\eqref{eq:asympt_psipm_n_infpos}.
 Hence, since $u_{0}=0$, $u$ is an eigenvector of the Jacobi operator $J^{(n)}$, whose matrix entries are given by the equations
 \[
  \left(J^{(n)}\right)_{k,k+1}=\left(J^{(n)}\right)_{k+1,k}=\sqrt{\gamma_{n+k}\gamma_{n+k+1}} \quad \mbox{ and } \quad \left(J^{(n)}\right)_{k,k}=0, \quad k\in\N,
 \]
 to the eigenvalue $\upsilon(z_{0})$. Since $J^{(n)}$ is Hermitian, $\upsilon(z_{0})\in\R$ which, together with $|z_{0}|<1$, implies that $z_{0}\in\R$.
 
 Second, we verify that any zero of $\psi_{n}^{+}$ has to be simple. For a contradiction, suppose $z_{0}\in\mathbb{D}\setminus\{0\}$ is a multiple zero of $\psi_{n}^{+}$.
 Then we already know that $-1<z_{0}<1$ and, moreover, we have
 \[
  \psi_{n}^{+}(z_{0})=\left(\psi_{n}^{+}\right)'(z_{0})=0.
 \]
 On the other hand, for any $x\in(-1,1)\setminus\{0\}$ and $n\geq-1$, it follows from the fact that $\psi^{+}_{n}(x)\in\R$,
 the asymptotic formula~\eqref{eq:asympt_psipm_n_infpos}, and Lemma~\ref{lem:aux_green_id} that
 \begin{equation}
  \left(\psi_{n}^{+}\right)'(x)\psi_{n+1}^{+}(x)-\psi_{n}^{+}(x)\left(\psi_{n+1}^{+}\right)'(x)<0.
  \label{eq:aux_id_wronsk_positiv}
 \end{equation}
 However, for $x=z_{0}$, the left-hand side of~\eqref{eq:aux_id_wronsk_positiv} vanishes which is a contradiction. 
 
 Since $z\mapsto z^{-n}\psi_{n}^{+}(z)$ is an even function, zeros of $\psi_{n}^{+}$ are distributed symmetrically around the origin.
 
 Finally, we prove the statement concerning the interlacing property of the zeros. Let $x_{0}$ and $x_{1}$, 
 $0<x_{0}<x_{1}<1$, be two consecutive zeros of $\psi_{n}^{+}$. Since these zeros are simple, one has
 \[
  \left(\psi_{n}^{+}\right)'(x_{0})\left(\psi_{n}^{+}\right)'(x_{1})<0.
 \]
 At the same time, according to~\eqref{eq:aux_id_wronsk_positiv}, one gets the inequalities
 \[
  \left(\psi_{n}^{+}\right)'(x_{0})\psi_{n+1}^{+}(x_{0})<0 \quad \mbox{ and } \quad \left(\psi_{n}^{+}\right)'(x_{1})\psi_{n+1}^{+}(x_{1})<0.
 \]
 These three inequalities imply that
 \[
  \psi_{n+1}^{+}(x_{0})\psi_{n+1}^{+}(x_{1})<0
 \]
 and hence there is at least one zero of $\psi_{n+1}^{+}$ between $x_{0}$ and $x_{1}$. However, there is exactly one zero of $\psi_{n+1}^{+}$ between $x_{0}$ and $x_{1}$. Indeed, assuming that there are at least two
 zeros of $\psi_{n+1}^{+}$ between $x_{0}$ and $x_{1}$, there are two consecutive zeros $y_{0}$ and $y_{1}$ of $\psi_{n+1}^{+}$ such that $x_{0}< y_{0}<y_{1}< x_{1}$ (note that $\psi_{n}^{+}$ and $\psi_{n+1}^{+}$ have
 no zero in common in $(0,1)$ since it would contradict~\eqref{eq:aux_id_wronsk_positiv}). By simplicity of $y_{0}$ and $y_{1}$, one has 
 \[
  \left(\psi_{n+1}^{+}\right)'(y_{0})\left(\psi_{n+1}^{+}\right)'(y_{1})<0.
 \]
 This, however, yields a contradiction with~\eqref{eq:aux_id_wronsk_positiv} since the function $\psi_{n}^{+}$ does not change sign in $(x_{0},x_{1})$.
\end{proof}

\bibliographystyle{acm}

\end{document}